\documentclass[acmsmall,nonacm]{acmart}

\AtBeginDocument{%
  \providecommand\BibTeX{{%
    \normalfont B\kern-0.5em{\scshape i\kern-0.25em b}\kern-0.8em\TeX}}}

\setcopyright{none}
\copyrightyear{2021}
\acmDOI{}

\acmPrice{}
\acmISBN{}

\usepackage[utf8]{inputenc}
\usepackage{booktabs}   %% For formal tables:
\usepackage{subcaption} %% For complex figures with subfigures/subcaptions
\usepackage{etoolbox}
\usepackage{amsmath}
\usepackage{amsthm}
\usepackage{listings}
\usepackage{multirow}
\usepackage{pifont}

\newcommand{\ssecRef}[1]{Subsection~\ref{#1}}
\newcommand{\secRef}[1]{Section~\ref{#1}}

\newcommand{\defRef}[1]{Definition~\ref{#1}}
\newcommand{\figRef}[1]{Figure~\ref{#1}}
\newcommand{\lstRef}[1]{Listing~\ref{#1}}
\newcommand{\tabRef}[1]{Table~\ref{#1}}

\newcommand{\cmark}{\ding{51}}

\usepackage{color}
\definecolor{mygray}{rgb}{0.4,0.4,0.4}
\definecolor{mygreen}{rgb}{0.2,0.6,0.2}
\definecolor{myorange}{rgb}{0.6,0.3,0}
\definecolor{myblue}{rgb}{0,0.4,1.0}
\definecolor{red4}{rgb}{0.4,0.1,0.1}
\definecolor{red3}{rgb}{0.6,0.2,0.2}
\definecolor{red2}{rgb}{0.8,0.3,0.3}
\definecolor{gray4}{rgb}{0.2,0.2,0.2}
\definecolor{gray3}{rgb}{0.4,0.4,0.4}
\definecolor{gray2}{rgb}{0.6,0.6,0.6}

\lstdefinelanguage{newc}{
  language=C,
  basicstyle=\footnotesize\ttfamily\color{black},
  morekeywords={[2]bool,int8_t,uint8_t,int16_t,uint16_t,int32_t,uint32_t,int64_t,uint64_t,size_t},
  sensitive=true,
  keywordstyle=\color{mygreen},
  keywordstyle={[2]\color{red4}},
  emph={__builtin_opacify,__builtin_observe_mem,__builtin_io},
  emphstyle=\color{mygray},
  commentstyle=\color{myblue},
  stringstyle=\color{mygreen},
  morecomment=[l][\sffamily\bfseries]{///},
  morestring=[b]{"},
  frame=single,
  framerule=0pt,
  aboveskip=3pt,
  belowskip=0pt,
  framesep=1pt,
  numbers=left,
  numbersep=5pt,
  numberstyle=\tiny\color{mygray},
  showspaces=false,
  showstringspaces=false,
  tabsize=2
}

\lstdefinelanguage{miniir}{
  basicstyle=\footnotesize\ttfamily\color{black},
  morekeywords={[1]br,call,function,macro,return,use,mem,io},
  morekeywords={[2]snapshot,opaque,yield},
  morekeywords={[3]tailio,token},
  sensitive=true,
  morecomment=[l]{///},
  morestring=[b]{"},
  commentstyle=\color{myblue},
  stringstyle=\color{mygreen},
  morecomment=[l][\sffamily]{///},
  frame=single,
  framerule=0pt,
  aboveskip=3pt,
  belowskip=0pt,
  framesep=1pt,
  numbers=left,
  numbersep=5pt,
  numberstyle=\tiny\color{mygray},
  keywordstyle={[1]\color{gray4}\bfseries},
  keywordstyle={[2]\color{gray3}\bfseries},
  keywordstyle={[3]\color{gray2}\bfseries},
  showspaces=false,
  showstringspaces=false,
  tabsize=2
}

\newcommand{\ic}[1]{\textrm{\textup{\lstinline[basicstyle=\ttfamily\color{black},stringstyle=\color{black},commentstyle=\color{black},keywordstyle={[1]\color{black}},keywordstyle={[2]\color{black}},keywordstyle={[3]\color{black}}]{#1}}}}
\newcommand{\mathic}[1]{\text{\textrm{\textup{\ttfamily\color{black}{#1}}}}}
\newcommand{\indextt}[1]{\index{#1@\protect\keyword{#1}}}

\newif\ifspace
\newif\ifnewentry
\newcommand{\addspace}{\ifspace ~ \spacefalse \fi}

\newcommand{\nonterm}[2]{%
  \ifblank{#2}%
   {\addspace\hbox{\textsl{#1}\ifnewentry\index{grammar entries!\textsl{#1}}\fi}\spacetrue}%
   {\addspace\hbox{\textsl{#1}\footnote{#2}\ifnewentry\index{grammar entries!\textsl{#1}}\fi}\spacetrue}}
\newcommand{\repetstar}{$^*$\spacetrue}
\newcommand{\repetplus}{$^+$\spacetrue}

\newcommand{\sep}{ \\[2mm] \spacefalse\newentrytrue}
\newcommand{\newl}{ \\ & & \spacefalse}
\newcommand{\alt}{ \\ & $\mid$ & \spacefalse}
\newcommand{\is}{ & $::=$ & \newentryfalse}

\newcommand{\COMP}[2]{\mathcal{C}[\hskip-2.5pt[#1]\hskip-2.1pt](#2)}
\newcommand{\EXEC}[2]{\mathcal{E}[\hskip-2.5pt[#1]\hskip-2.1pt](#2)}

\newcommand{\INST}{\mathit{Inst}}
\newcommand{\INSTLIST}{\mathit{InstList}}

\newcommand{\LEADSTO}[1]{\stackrel{\mathrm{#1}}{\leadsto}}
\newcommand{\OPAQUETO}{\stackrel{\scriptscriptstyle\mathrm{\!opaque\!}}{\leadsto}}

\newcommand{\TRANS}{\tau}
\newcommand{\TRANSMAP}[1][\tau]{\mathrel{\propto_{#1}}}

\newcommand{\HBREL}[2][]{\mathrel{\stackrel{\mathrm{\scriptscriptstyle#2}}{\rightarrow}\!\!^{#1}}}
\newcommand{\HB}[1][]{\HBREL[#1]{hb}}
\newcommand{\HBIO}[1][]{\HBREL[#1]{io}}
\newcommand{\HBDU}[1][]{\HBREL[#1]{du}}
\newcommand{\HBRF}[1][]{\HBREL[#1]{rf}}
\newcommand{\HBCD}[1][]{\HBREL[#1]{cd}}
\newcommand{\HBDEP}[1][]{\HBREL[#1]{dep}}
\newcommand{\HBOF}[1][]{\HBREL[#1]{of}}
\newcommand{\HBOO}[1][]{\HBREL[#1]{oo}}

\newsavebox{\fmbox}
\newenvironment{cadre}
     {\begin{lrbox}{\fmbox}\begin{minipage}{1.01\textwidth}}
     {\end{minipage}\end{lrbox}\fbox{\usebox{\fmbox}}}

\begin{document}

%\title{Reconciling Observation With Optimization in Secure Compilation}
\title{Secure Optimization Through Opaque Observations}

\author{Son Tuan Vu}
\affiliation{%
  \institution{Sorbonne Universit\'e, CNRS, LIP6}
  \streetaddress{4 place Jussieu}
  \postcode{75252}
  \city{Paris}
  \country{France}}
\email{son-tuan.vu@lip6.fr}
\author{Albert Cohen}
\affiliation{%
  \institution{Google}
  \city{Paris}
  \country{France}}
\email{albertcohen@google.com}

\author{Karine Heydemann}
\affiliation{%
  \institution{Sorbonne Universit\'e, CNRS, LIP6}
  \streetaddress{4 place Jussieu}
  \postcode{75252}
  \city{Paris}
  \country{France}}
\email{karine.heydemann@lip6.fr}

\author{Arnaud de Grandmaison}
\affiliation{%
  \institution{Arm}
  \city{Paris}
  \country{France}}
\email{arnaud.degrandmaison@arm.com}

\author{Christophe Guillon}
\affiliation{%
  \institution{STMicroelectronics}
  \city{Grenoble}
  \country{France}}
\email{christophe.guillon@st.com}

\renewcommand{\shortauthors}{S. Vu et al.}

\begin{abstract}
  Secure applications implement software protections against side-channel and physical attacks. Such protections are meaningful at machine code or micro-architectural level, but they typically do not carry observable semantics at source level. To prevent optimizing compilers from altering the protection, security engineers embed input/output side-effects into the protection. These side-effects are error-prone and compiler-dependent, and the current practice involves analyzing the generated machine code to make sure security or privacy properties are still enforced.  Vu et al.\ recently demonstrated how to automate the insertion of volatile side-effects in a compiler \cite{Vu20}, but these may be too expensive in fined-grained protections such as control-flow integrity. We introduce observations of the program state that are intrinsic to the correct execution of security protections, along with means to specify and preserve observations across the compilation flow. Such observations complement the traditional input/output-preservation contract of compilers. We show how to guarantee their preservation without modifying compilation passes and with as little performance impact as possible. We validate our approach on a range of benchmarks, expressing the secure compilation of these applications in terms of observations to be made at specific program points.
\end{abstract}




\maketitle

\renewcommand{\thefootnote}{\fnsymbol{footnote}}
\footnotetext[1]{\bfseries Preprint presented at the PriSC workshop, January 17, 2021 (with POPL 2021).}
\renewcommand{\thefootnote}{\arabic{footnote}}

\section{Introduction}

Compilers care about preserving the input/output (I/O) behavior of the program; they achieve this by preserving functional correctness of computations linking input to output values, and making sure these take place in the order and under the conditions specified in the source program.

Interestingly, this is not enough for many compilation scenarios, where optimizations are too aggressive at removing, reordering or otherwise modifying computations that do not result in externally visible
I/O. We identified four such scenarios:

\begin{enumerate}
  \item \textbf{Preventing software and side-channel attacks.} Optimizing compilers are known to interfere with a wide range of security properties. For example, dead store elimination may optimize out procedures specifically designed to erase sensitive data in memory, thereby exposing encryption keys and other secrets to be accessed by an attacker or captured in a memory dump \cite{silva, simon, erasure}. In cryptography applications, countermeasures against side-channel attacks are even more fragile: optimizations may invalidate a wide range of masking protections---randomizing sensitive data \cite{rivain}---by reordering elementary operations in masking expressions \cite{bayrak}. Related to this and to forestall timing attacks, encryption/decryption kernels are usually designed to run for a constant amount of time, independent of sensitive inputs \cite{cachebleed}. To achieve this, security engineers go to great lenghts to write straight line code, carefully avoiding control flow depending on sensitive data \cite{simon}; unfortunately, compilers often substitute dataflow encodings of control flow (conditional moves, bitwise logic) with more efficient, performance-wise or size-wise, conditional control flow, defeating the purpose of the constant time implementation \cite{simon}.
  
  \item \textbf{Preventing fault attacks.} Fault attacks are a growing threat for embedded systems. They can alter the system's correct behavior by means of physical interference \cite{yuce}. Software countermeasures against such attacks usually involve some form of redundancy, such as computing the same operation on one or more copies of the data then comparing the results \cite{barel, hillebold, proy, barry}. However, this has been a constant fight with compilers as one of the most essential goals of compiler optimizations is to removing redundant code \cite{hillebold, simon}. Another well-known countermeasure which aims at detecting fault attacks altering the program's control flow consists in incrementing counters along side with the execution of individual instructions, source code statements, function calls, etc., then checking their expected values when reaching control-flow merge points \cite{lalande}. Once again, as compilers do not model the security intention of the countermeasure, they will remove the trivially true counter checks or collect multiple counter incrementations into a single addition.

  \item \textbf{Testing, inspecting or verifying machine code.} In security-sensitive applications, these are classical procedures, mandated by certification authorities and standards. It includes checking for the presence of countermeasures against attacks of the form described in the previous two items. There has been a large body of work in software security showing the importance of analysis and verification tools assessing program properties---expressed as a propositional logic formula---at the level of machine code \cite{balakrishnan, brejon, angr}. The need for such analyses derives from the observable mismatch between the behavior intended by the programmer and what is actually executed by the processor \cite{balakrishnan}, or from the essential role played by low-level (micro-)architectural features \cite{brejon}. More generally, machine code analysis tools and runtime monitors (including debuggers) often need program properties representing high-level program specification, from which the static or dynamic analysis may determine whether the application is secure against a given attack \cite{brejon}. Still, it is generally challenging to propagate source-level program properties all the way down to machine code \cite{Vu20}. Compilers have no notion of the link between properties expressed as annotations, semantical comments or pragmas, and the semantics of the code they refer to. As a result, compilers generally do not preserve this link or update properties according to the transformations they apply to the code (e.g., updating static bounds on the number of loop iterations when performing loop unrolling). Besides, variables referenced in program properties may also be affected by compiler optimizations, e.g.\ such variables may be optimized out, thus invalidating the property \cite{Vu20}.

\item \textbf{Debugging, unit testing.} Source-level properties can be
  a powerful debugging tool, helping enforce or detect security
  violations through the development and deployment process. This is
  similar to inserting (runtime) assertions, except that the code does
  not have to carry these runtime checks when running in production. Instead, the ability to propagate  source-level properties down to machine code, allows a debugger to trace program execution
  %%\TODO{trace into ? trace the program ?} AC: trace the program, you're right
  and evaluate the properties in a testing environment. This brings the promise of using the same executable code in both testing and production environments. Unlike assertions, it is expected that machine code annotations do not consume execution time. On the contrary, their preservation through the compilation flow is not easy (as we have seen earlier in secure applications) and may also force the compiler to occasionally restrain its full optimization potential. We are not aware of any software engineering approach being currently pursued together with an aggressively optimizing compiler, precisely because optimizations are prone to destroying the link between the high-level properties and the machine code where they are meant to be tested.  And indeed, while these are common issues related to the propagation of debug information through optimization passes, the usual tradeoff in such a case is to preserve debug information only as a best-effort strategy, which is insufficient for our debug-only-assert scenario.
\end{enumerate}

This paper is motivated by security applications, and we will focus on the first three scenarios in the following. The effectiveness of our approach in broader software engineering areas such as the fourth scenario is left for future work.

The common pattern in the above scenarios is that program transformations---such as abstraction lowering steps and optimizations---have no means to reason about non-I/O-related observations that the
programmer would like to specify and whose preservation should be enforced. As a result, security engineers resort to embedding I/O such as volatile side-effects into security protections. Vu et al.\ recently demonstrated how to automate this in a compiler \cite{Vu20}; but volatile side-effects may be too expensive. We propose a means to \emph{specify and preserve observations}, to do so in \emph{the most general manner}, avoiding to modify each and every transformation pass in a compiler, extending automatically to future transformations, and with as little
performance impact as possible. To this end, we leverage the most essential information modeled by nearly every compiler transformation: \emph{I/O effects and data flow}. And we control this information
according to the specified observations through the ability to hide information about an atom of operational semantics---a.k.a.\ \emph{opacity}. We provide a concrete syntax and formal semantics to specify
observations and preserve them through the compilation flow. We show that preserving observations does not require significant modifications to compilers and demonstrate this on the LLVM compiler with
aggressive optimizations turned on. We validate our approach and implementation on a range of security-sensitive benchmarks, expressing the secure compilation of these applications in terms of observations
to be made at specific points of the computation.

In this paper, we contribute the following:
\begin{itemize}
  \item We define the notion of observation and its preservation through program transformations (\secRef{sec:problem}).
  \item We present our mechanism to preserve observations down to machine code with minimal interference with compiler optimizations, and formally prove its correctness using a simplified intermediate
  program representation (\secRef{sec:proposed}).
  \item We detail our LLVM-based implementation with
    %%virtually\TODO{virtually ???} AC: almost, but virtually works as well
    almost no modification to individual optimization passes (\secRef{sec:implementation}).
  \item We study concrete applications of our approach, to preserve
    security protections introduced at source level (\secRef{sec:security}).
  \item We validate the preservation of these security properties according to a range of criteria, to establish the correctness of our implementation (\secRef{sec:validation}).
  \item We evaluate the performance and compilation time impact of our approach, and further compare with alternative mechanisms (\secRef{sec:evaluation}).
\end{itemize}

\section{Related Work}
\label{sec:related}

There is a large body of research and engineering on secure compilation
\cite{abadi98,abadi2012,devriese,gorla_nestmann_2016,DBLP:journals/corr/abs-1807-04603}.
The correctness of a compiler is defined w.r.t.\ to a notion of behavioral
equivalence, which may take different forms from full abstraction to more
specific type and isolation properties
\cite{abadi2012,Chlipala:2007:CTC:1250734.1250742,Patrignani:2015:SCP:2764452.2699503}
or even hyperproperties not directly captured in terms of behavioral equivalence
\cite{DBLP:conf/csfw/AbateB0HPT19}.  Behavioral equivalence is generally defined
with respect to the capabilities of an attacker.  Our secure compilation problem
targets a wide class of properties not directly expressible within a source
language semantics: properties of the machine state resulting from the
compilation of the program, including logical properties of side-channels and
countermeasures to physical attacks. Like Blazy et al.\
\cite{DBLP:journals/pacmpl/BartheBGHLPT20} we thus extend the semantics of a
host language with state and denotations to reason about these extra-functional
properties. But unlike Blazy et al.\ we do not focus on a specific kind of
property (execution time in their case). We provide a general means to express
observations at deterministic points of the execution, even though the program
is subject to aggressive transformations. As presented by Vu et al.\
\cite{Vu20}, these observations may then enable the expression and validation of
a wide range of logical properties \cite{ACSL}.

This section discusses the most closely related work, starting with Vu et al.\
\cite{Vu20} which laid the groundwork for this paper. To prevent interference
from compiler optimizations, security engineers resort to embedding I/O such as
volatile side-effects into security protections.  Vu et al.\ automated this
process \cite{Vu20}, but as
%% we have shown\TODO{??? discussed ? explained ? other ?} in the previous section, side-effects may be too. AC: stale reference to the evaluation, when related work was in the end; same for some TODOs below (fixed and removed)
our experiments will show, side-effects may be too expensive in scenarios such
as fine-grained control-flow integrity.  This fact motivates our effort to
distinguish observations from regular I/O mechanisms, and not encoding
observations as fake I/O instructions.  In addition, Vu et al.\ relied on a
restrictive notion of behavioral equivalence by enforcing the equality of I/O
and observation traces. In this paper, we provide security engineers with
finer-grained control on the preservation of observations across
transformations, and on the partial ordering of these observations.

Compilers for hard real-time systems are designed to carry detailed control flow
information in order to bound the worst-case execution time of a reactive method
as accurately as possible. This information is called \textit{flow information}
and takes the form of source code annotations about, for instance, loop trip
counts, infeasible paths and program points that are mutually exclusive during a
given run \cite{10.1007/978-3-642-16256-5_6,aiT}. There is no attempt at
formalizing the preservation of control-flow information as a correctness
requirement. Instead, CompCert relies on known and implementation-specific
limitations of the compiler: it introduces a builtin function modeled as a call
to an external function producing an observable event, without emitting it as
machine code \cite{compcert}.
 
The ENTRA (Whole-Systems ENergy TRAnsparency) project Deliverable D2.1
\cite{entra} describes a similar mechanism to transfer information from source
to machine code. Data and control flow properties are encoded as comments
written as inline assembly expressions, relying on the compiler to preserve the
local variables listed in the assembly expression. These expressions are
declared as volatile I/O side-effecting to maintain their position in control
flow relative to other code. This mechanism can be used to observe values and
preserve them, but cannot be used to chain these observations and implement a
partial ordering specification as the properties do not produce any value (i.e.\
no opacification). As a result, they cannot be used to preserve security
protections.

Another safety-minded approach encodes flow information using IR extensions and
external transformations to update loop trip count information \cite{li}. This
approach incurs significant changes to optimization passes: it comes with a set
of rules to transform control flow information along code transformations.

The introduction already listed applications and motivating scenarios in
security and software engineering. The reader may refer to Vu et al.\
\cite{Vu20} for a more extensive discussion of security-related work.

\section{Motivating Example}
\label{sec:example}

Let us consider a cryptography application as a motivating example. The leakage of confidential information such as secret keys is a major threat. A common countermeasure consists in erasing sensitive data from memory once they are no longer needed \cite{erasure}, including keys, seeds of random generators, and temporary encryption or decryption buffers.

However, this may not be as easy as it seems: security engineers have
been painfully fighting optimizing compilers to achieve their goal
\cite{erasure, simon}. Consider the example in
\lstRef{lst:example}. The \ic{secret} buffer containing sensitive
information is allocated on the stack and should be erased before
returning from the function; this is implemented through a call to
\ic{memset()}. However, compilers will spot that \ic{secret} goes out
of scope, meaning that access after the function returns is an error or has unspecified behavior; since many optimizing compilers model \ic{memset()} as a builtin function, they are aware of its semantics and will consider the call as dead stores, removing the erasure as part of ``dead store elimination''. While this completely subverts the security protection, it is perfectly correct with regard to the C semantics: the observable effects of the program are not modified by this optimization.

\begin{figure}[h!tb]
  \begin{lstlisting}[xleftmargin=0.35cm, caption={Erasing a secret buffer on the stack.},
                     captionpos=b, language=newc, label={lst:example}]
void process_sensitive(void) {
  uint8_t secret[32];
  ...
  memset(secret, 0, sizeof(secret));
}
  \end{lstlisting}
\end{figure}

Different solutions have been implemented in OpenSSL \cite{openssl} and mbedTLS \cite{mbedtls}, however none of them is guaranteed to work in all cases: compilers might still recognize the tricks and optimize them away \cite{simon, erasure}, thus potentially allowing unauthorized access to sensitive data.

At this point, we would like to emphasize that the security protection
can be equivalently expressed as a write of the zero value to the
secret buffer. This hints at a more general challenge of preserving
specific values at specific points of the program execution. These
points and values are associated with protection schemes and
countermeasures dictated by security concerns, and they have to be
preserved and traced down to machine code. As shown by this motivating
example, applications are commonly secured by inserting protections at
source level, but compilers may fail to implement the programmers'
intentions as these protections often do not alter the program
observable behavior, resulting in unsafe machine code. Our work aims
to enable the programmer to instruct compilers into preserving
specific computations and values, thus enforcing the associated
security protections and properties. In the next section, we will
formally define this observation specification mechanism and the value
preservation problem.

\section{Problem Statement and Definitions}
\label{sec:problem}

\lstset{
  language=miniir
}

We introduce a simple language, called Mini IR, representative of the levels of Intermediate Representation (IR) typical of the compilation of imperative languages. It models control flow at a relatively low level---linearized three-address code--- while supporting the usual memory abstractions, SSA values, intra- and inter-procedural constructs. For the sake of simplicity, we will use it to model not only optimizing compiler IR but also source code and assembly code.

\subsection{Mini IR Syntax}

Figure~\ref{fig:gram} presents the grammar of Mini IR (it will be extended in the next section).

\begin{figure}[h!tb]
  \begin{cadre}
    \footnotesize
    \newenvironment{syntax}{\begin{tabular}{@{}rrll@{}}\spacefalse}{\end{tabular}}

\begin{syntax}
  \nonterm{const}{} \is{} \nonterm{integer}{} & \textrm{ integer constants} \sep{}  \nonterm{var}{} \is{} \nonterm{ident}{} & \textrm{ identifier for an SSA value} \sep{}  \nonterm{expr}{} \is{} \nonterm{const}{} & \textrm{} \alt{} \nonterm{un-op}{} \nonterm{var}{} & \textrm{ unary operator} \alt{} \nonterm{var}{} \nonterm{bin-op}{} \nonterm{var}{} & \textrm{ binary operator} \alt{} \nonterm{ident}{} \addspace\lstinline[keywordstyle={}]$($\spacetrue \nonterm{var}{}\repetstar{} \addspace\lstinline[keywordstyle={}]$)$\spacetrue & \textrm{ function application or macro expansion} \alt{} \addspace\lstinline[keywordstyle={}]$io($\indextt{io}\spacetrue \nonterm{var}{}, \nonterm{var}{}\repetstar{} \addspace\lstinline[keywordstyle={}]$)$\spacetrue & \textrm{ I/O effect with ordering descriptor} \alt{} \addspace\lstinline[keywordstyle={}]$snapshot($\indextt{snapshot}\spacetrue \nonterm{var}{}\repetplus{} \addspace\lstinline[keywordstyle={}]$)$\spacetrue & \textrm{ identity function observing its arguments into a} \newl{}                              & \textrm{ partial state} \sep{}  \nonterm{instr}{} \is{} \nonterm{expr}{} & \textrm{ expression with no associated definition} \alt{} [ \nonterm{var}{}\repetplus{} \addspace\lstinline[keywordstyle={}]$=$\spacetrue ] \nonterm{expr}{} & \textrm{ define a value from an expression} \alt{} \nonterm{ref}{} \addspace\lstinline[keywordstyle={}]$<-$\spacetrue \nonterm{var}{} & \textrm{ store a value to a reference} \alt{} \nonterm{var}{} \addspace\lstinline[keywordstyle={}]$=$\spacetrue \nonterm{ref}{} & \textrm{ load a value from a reference} \alt{} \addspace\lstinline[keywordstyle={}]$mem[$\indextt{mem}\spacetrue \nonterm{var}{} \addspace\lstinline[keywordstyle={}]$] <-$\spacetrue \nonterm{var}{} & \textrm{ store a value to a memory address} \alt{} \nonterm{var}{} \addspace\lstinline[keywordstyle={}]$= mem[$\indextt{mem}\spacetrue \nonterm{var}{} \addspace\lstinline[keywordstyle={}]$]$\spacetrue & \textrm{ load a value from a memory address} \alt{} \addspace\lstinline[keywordstyle={}]$br$\indextt{br}\spacetrue \nonterm{var}{}\addspace\lstinline[keywordstyle={}]$,$\spacetrue \nonterm{ident}{} \addspace\lstinline[keywordstyle={}]$($\spacetrue \nonterm{var}{}\repetstar{} \addspace\lstinline[keywordstyle={}]$)$\spacetrue & \textrm{ branch with condition, target block identifier and} \newl{}                                        & \textrm{ arguments} \alt{} \addspace\lstinline[keywordstyle={}]$return($\indextt{return}\spacetrue \nonterm{var}{}\repetstar{} \addspace\lstinline[keywordstyle={}]$)$\spacetrue & \textrm{ return from function or macro} \sep{}  \nonterm{block}{} \is{} \nonterm{ident}{} [ \addspace\lstinline[keywordstyle={}]$($\spacetrue \nonterm{var}{}\repetstar{} \addspace\lstinline[keywordstyle={}]$)$\spacetrue ] \addspace\lstinline[keywordstyle={}]$:$\spacetrue [ \nonterm{instr}{} \addspace\lstinline[keywordstyle={}]$;$\spacetrue ]\repetstar{} & \textrm{ block labeled by a unique identifier,} \newl{}                                                      & \textrm{ composed of arguments and an instruction sequence,} \newl{}                                                      & \textrm{ branch- or return-terminated} \sep{}  \nonterm{region}{} \is{} \addspace\lstinline[keywordstyle={}]${$\spacetrue \nonterm{block}{}\repetplus{} \addspace\lstinline[keywordstyle={}]$}$\spacetrue & \textrm{ return-terminated region with one or more blocks} \sep{}  \nonterm{func-decl}{} \is{} \addspace\lstinline[keywordstyle={}]$function$\indextt{function}\spacetrue \nonterm{ident}{} \addspace\lstinline[keywordstyle={}]$($\spacetrue \nonterm{var}{}\repetstar{} \addspace\lstinline[keywordstyle={}]$)$\spacetrue \nonterm{region}{} & \textrm{ function definition} \sep{}  \nonterm{macro-decl}{} \is{} \addspace\lstinline[keywordstyle={}]$macro$\indextt{macro}\spacetrue \nonterm{ident}{} \addspace\lstinline[keywordstyle={}]$($\spacetrue \nonterm{var}{}\repetstar{} \addspace\lstinline[keywordstyle={}]$)$\spacetrue \nonterm{region}{} & \textrm{ macro definition} \sep{}\end{syntax}

       %% | "if" var "then" region ["else" region ] ; conditional control flow
       %% | "while" var region ; iteration control flow

       %% | "atomic" region ; make I/O, side-effects and definitions visible
       %%                   ; atomically outside the region and vice versa

%%% Local Variables:
%%% mode: latex
%%% TeX-master: "main"
%%% End:

    \normalsize
  \end{cadre}
  \caption{Grammar of our Mini IR. The terminals \emph{ident}, \emph{un-op}, \emph{bin-op}, \emph{integer} are the same as the corresponding C lexical tokens.}
\label{fig:gram}
\end{figure}

In Mini IR, control flow is implemented as flat Control Flow Graph (CFG) of blocks and branches. Unlike traditional CFG- and SSA-based compilers (GCC and LLVM), we use branch and block arguments following
continuation-passing style \cite{DBLP:journals/sigplan/Appel88}. Like in MLIR \cite{lattner2020mlir}, single-assignment variables are declared and scoped in a region (introduced by a function or macro) and captured in dominated blocks; as a consequence, branch arguments only need to carry SSA values, as opposed to explicitly carrying all live variables. This choice makes use-def chains more uniform across an entire function without implicitly referring to control flow edges, which in
turns simplifies our formalization of a happens-before relation later in this section.

Note that we did not include indirect branches and calls in the syntax. Supporting these would include making identifiers first class and holding them as additional SSA arguments of branch and call instructions. This does not impact the following formalization and expressing our secure compilation benchmarks.

When clear from the context, we will write ``instruction $\mathit{expr}$'' when referring to a instruction defining, assigning or returning a value from an expression $\mathit{expr}$.

\subsection{Operational Semantics}

All expressions and instructions have fairly standard semantics, except for \ic{snapshot}, which will be presented in the next subsection.

As a simplifying assumption, we only consider \emph{sequential}, \emph{deterministic} programs with \emph{well defined behavior}. In particular, we avoid cases where the compiler may take advantage of
undefined behavior to trigger optimizations. This assumption is consistent with widespread coding standards for secure code. Our formalization also assumes \emph{no exceptions} at the source language level, but precise machine exceptions at the instruction level are supported. % but we will discuss this hypothesis in \secRef{sec:implementation}.

%By doing so, we do not limit
%ourselves to the optimizations currently available in LLVM or any C
%compiler, but rely on the more fundamental assumption that
%optimizations rely on static data-flow properties to validate
%a given program transformation. We believe our problem
%statement and solutions could be adapted to other source languages and
%compilers with only minor modifications. In the following, the program
%may refer to any step of the compilation flow: source, IR or machine
%code level.
\begin{definition}[Name-value domains]
  Every value manipulated during the execution of a program belongs to one of these four \emph{Name-value domains}:
  \begin{itemize}
    \item $\mathcal{V}$ is a set of $(\mathit{Var}, \mathit{Val})$ pairs where $\mathit{Var}$ is an SSA variable name (e.g.\ an SSA value in LLVM IR or a variable in a functional language) and $\mathit{Val}$ is the value
    of $\mathit{Var}$; all uses of $\mathit{Var}$ are dominated by an unique definition associating $\mathit{Var}$ with its value $\mathit{Val}$;
    \item $\mathcal{C}$ is a set of $(\mathit{Val}, \mathit{Val})$ pairs where $\mathit{Val}$ is a constant value also standing as the name of the constant;
    \item $\mathcal{R}$ is a set of $(\mathit{Ref}, \mathit{Val})$ pairs where $\mathit{Ref}$ is a reference name (e.g.\ a C variable, a reference in a functional language, or a register in a low-level representation) and $\mathit{Val}$ is the value referenced by $\mathit{Ref}$;
    \item $\mathcal{M}$ is a set of $(\mathit{Mem}, \mathit{Val})$ pairs where $\mathit{Mem}$ is a memory address and $\mathit{Val}$ is the value stored at $\mathit{Mem}$.
  \end{itemize}
\end{definition}

We define an operational semantics for our Mini IR in terms of a state machine,
where every IR instruction defines a transition referred to as an \emph{event}.

\begin{definition}[Program state]
  A \emph{program state} is defined by a tuple $(\mathit{Vals}, \pi)$ with $\mathit{Vals} = V \cup C \cup R \cup M$, where $\mathit{V} \subseteq \mathcal{V}$, $\mathit{C} \subseteq \mathcal{C}$, $\mathit{R} \subseteq \mathcal{R}$, $\mathit{M} \subseteq \mathcal{M}$, and the \emph{program point} $\pi$ holds the value of the program counter pointing to the next instruction to be executed.
\end{definition}

\begin{definition}[Event]
  An \emph{event} $e$ is a state machine transition, associated with the execution of an instruction $i$, from a state $\sigma$ into a state $\sigma'$. It is denoted by $e = \sigma \LEADSTO{i} \sigma'$.
\end{definition}

%% AC: Deprecated text for structured control flow
%% structured control flow expressions---\ic{if-then-else} and \ic{while}---produces an event associated with the evaluation of their condition, then jumps into the appropriate enclosed region, instanciating an event for every
%% instruction in the region, and exiting the region when encountering a \ic{yield} instruction. The execution optionally iterates on the region in the case of a \ic{while} loop. Control flow eventually
%% continues with the following instruction, defining all values listed in the \ic{yield} instruction, unless internal control flow does not terminate.

For any given event $e$, let $\INST(e)$ denote the instruction executed by event $e$.

\begin{definition}[Program execution]
  A program execution $E$ is a---potentially infinite---ordered sequence of program states and events:
  \begin{align*}
    E = \sigma_0e_0\sigma_1e_1\sigma_2\ldots \textrm{ with } {} & e_0 \text{ a special initial event defining all constant values } c \in \mathcal{C}, \\
    {} & \sigma_0 \textrm{ the initial state, and } \sigma_k \LEADSTO{i_k} \sigma_{k + 1}, \textrm{ where } \forall k \geq 0, i_k=\INST(e_k)
  \end{align*}
\end{definition}

When executing a non-branch, non-\ic{return} instruction in a basic block,
the next state $\sigma_{k+1}$ points to the next instruction in the
block. Executing a branch instruction makes the next state point to the first
instruction of the target block. Executing a \ic{return} instruction makes
the next state point to the next instruction following the function call
instruction that led to the currently executing function.
When executing a function call instruction, the next state $\sigma_{k+1}$ points
to the first instruction of the function's enclosed region.

Unlike a function call, macro expansion takes place in an earlier, offline stage, prior to program execution. The macro's region is expanded in place, with effective arguments substituted in place of the formal ones, renaming the region's local variables and references to avoid conflicts with variables and references of the parent region, and implementing \ic{return} as copying some of the macro's variables into variables defined in the parent. As a result, macro expansion never occurs on a program execution.

Starting from an initial state ${\sigma}_0$, the execution proceeds with calling the special \ic{main} function, taking no argument and returning
no value. Instead the program conducts input and output operations through I/O instructions involving the \ic{io} expression. A given program input and output is modeled as
a---potentially infinite---list of independent---potentially infinite---sets of
values, each set identified with an unique descriptor, the first argument of the \ic{io} expression. Every value in an I/O set is uniquely tagged to distinguish it from any other I/O value from the same set.

The semantics of $P$ is a function from input sets to outputs sets. Given an input $I$, the semantics of $P$ applied to $I$ is denoted by $\COMP{P}{I}$, and $P$ produces an unique execution denoted by $\EXEC{P}{I}$. 

The execution of an I/O instruction instantiates an I/O event. Every I/O event
reads or writes one or more values. For an I/O event $e$ we note
$\mathit{IO}(e)$ its input or output values.

The sets $\mathit{Inputs}$ (resp.\ $\mathit{Outputs}$) represent the sets of all possible inputs of $P$ (resp.\ outputs of $P$), and $\mathit{Executions}$ is the set of executions produced by $P$.

\begin{definition}[I/O ordering]
  \label{def:io_ordering}
Any pair of distinct events associated with the execution of \ic{io} instructions \emph{with the same descriptor} are ordered by a so-called \emph{I/O ordering relation}, denoted by $\HBIO$.
Formally, given an execution $E = \EXEC{P}{I}$ of $P$ on some input $I$, $\HBIO$ is the reflexive and transitive closure of the following relation:
\begin{align*}
    \forall \ldots e_1 \ldots e_2 \ldots \in E,\ {} & \INST(e_1) = \text{\mathic{io}}(\mathit{desc}, \mathit{IO}(e_1))\\
    {} & \land\ \INST(e_2) = \text{\mathic{io}}(\mathit{desc}, \mathit{IO}(e_2)) \implies e_1 \HBIO e_2
\end{align*}
This relation on events induces a relation on input and output sets, also denoted by $\HBIO$
\begin{align*}
  \forall \ldots e_1 \ldots e_2 \ldots \in E,\ {} & \INST(e_1) = \text{\mathic{io}}(\mathit{desc}, \mathit{IO}(e_1))\\
    {} & \land\ \INST(e_2) = \text{\mathic{io}}(\mathit{desc}, \mathit{IO}(e_2))\ \land\ e_1 \HBIO e_2 \implies \mathit{IO}(e_1) \HBIO \mathit{IO}(e_2)
\end{align*}
In addition, when a single I/O event reads or writes multiple values, they are ordered from left to right in a given \ic{io} instruction and sequentially over successive \ic{io} expressions associated with the same event.

The $\HBIO$ relation on input and output data models streaming I/O as well as unordered persistent storage in computing systems, and any middle-ground situations such as locally unordered streaming I/O and locally ordered storage operations.
\end{definition}

\subsection{Program Transformations}

Let us first define a notion of program transformation, as general as possible, and without considering validity (correctness) issues for the moment. This notion is inseparable from a mapping that relates semantically connected events across program transformations.

\begin{definition}[Program transformation]
  Given a program $P$, a transformation $\TRANS$ maps $P$ to a transformed program $P'$. Every transformation $\TRANS$ induces an \emph{event map} $\TRANSMAP$ relating some events before and after transformation. The event map notation $e \TRANSMAP e'$ reads as ``$\TRANS$ maps $e$ to $e'$'' or ``$e$ maps to $e'$ through $\TRANS$'', or ``$e$ maps to $e'$'' when $\TRANS$ is clear from the context, or ``$\TRANS$ preserves $e$'' when the event after transformation does not need to be identified. The mapping is partial and neither injective nor surjective in general, as events in $P$ may not have a semantically relevant counterpart in $P'$ and vice versa.
\end{definition}

In the following, we will incrementally construct a $\TRANSMAP$ relation for an arbitrary transformation $\TRANS$. Being a constructive definition, it will serve as a tool to prove the existence, ordering, and properties of values across program transformations.

The set of hypotheses on what is considered a valid program transformation is minimal, covering as many compilation scenarios as possible. This constitutes a major strength of our proposal: we make no
assumptions on the analysis and transformation power of a compiler, covering not only the classical scalar, loop and inter-procedural transformations (optimization, canonicalization, lowering), but also hybrid static-dynamic schemes, including control and value speculation. The only constraint on transformations is to preserve the I/O behavior of a program \emph{on all possible inputs}.

\begin{definition}[Valid program transformation]
  \label{def:valid}
  Given a program $P$, a program transformation $\TRANS$ that applies to $P$ is valid if it produces a program $P'=\TRANS(P)$ such that $\forall I\in\mathit{Inputs}, \COMP{P}{I} = \COMP{P'}{I}$, i.e.\ $P$
  and $P'$ have the same I/O behavior.

  The set of all valid transformations of $P$ is denoted by $\mathcal{T}(P)$.
\end{definition}

Let us now prove that I/O events as well as their relative ordering are
preserved by all valid program transformations. We first introduce a class of
events that are always related through $\TRANSMAP$ for any valid transformation
$\TRANS$, then prove I/O events belong to this class.

\begin{definition}[Transformation-preserved event]
  \label{def:tpe}
  Given a program $P$ and input $I$, an event $e_{\mathit{tp}}$ is \emph{transformation-preserved} for execution $\EXEC{P}{I}$ if all valid program transformations are guaranteed to preserve it.
  The set of transformation-preserved events for a program $P$ and input $I$ is denoted by $\mathit{TP}(P,I)$.
  Formally,
  \begin{multline*}
  \forall e_{\mathit{tp}} \in \EXEC{P}{I},\ e_{\mathit{tp}} \text{ is a transformation-preserved event if and only if } \\
  \forall \TRANS \in \mathcal{T}(P),\ \exists e'_{\mathit{tp}} \in \EXEC{\TRANS(P)}{I}, e_{\mathit{tp}} \TRANSMAP e'_{\mathit{tp}}
  \end{multline*}
\end{definition}

Let us now show that one may construct a $\TRANSMAP$ relation that preserves I/O events.

\begin{lemma}[Unicity of transformed I/O events]
  \label{lem:utioe}
  For an execution $E = \EXEC{P}{I}$ of a program $P$ on some input $I$, an event $e$ from $E$ reading or writing a value $v$ from/to an input/output set, and a valid program transformation $\TRANS$, there exists a unique event $e' \in \EXEC{P'}{I}$ such that $e'$ reads or writes $v$.
\end{lemma}

\begin{proof}
  By definition of transformation validity (Definition~\ref{def:valid}), $v$ also belongs to an input or output set associated with the transformed program $P'=\TRANS(P)$. As a consequence, $E'=\EXEC{P'}{I}$ also holds an event $e'$ reading or writing $v$. Since $v$ is uniquely tagged among I/O values, semantical equality $\COMP{P}{I}=\COMP{P'}{I}$ implies that $e'$ is the only event reading or writing $v$ in the execution $E'$.
\end{proof}

\begin{definition}[Preservation of I/O events]
  \label{def:pioe}
  For an execution $E = \EXEC{P}{I}$ of a program $P$ on some input $I$ and a valid program transformation $\TRANS$, we define $\TRANSMAP$ to include all pairs $(e, e')$ such that $e$ is an I/O event in $E$ reading or writing a value $\mathit{val}$ from/to an input/output set, and $e'$ is the unique I/O event in $\EXEC{P'}{I}$ such that $e'$ reads or writes $\mathit{val}$.
\end{definition}

\begin{lemma}[Preservation of I/O event ordering]
  \label{lem:pioeo}
  Any valid program transformation preserves the partial ordering on I/O events.
\end{lemma}

\begin{proof}
  Consider the execution $E$ of a program $P$ on some input $I$, and a valid program transformation $\TRANS$.

  Given two events $e_1$ and $e_2$ in $E$, each of which is associated with an \ic{io} expression such that $e_1 \HBIO e_2$. From Definition~\ref{def:pioe}, there exists two events $e'_1$ and $e'_2$ in $E'=\EXEC{\TRANS(P)}{I}$ such that $e_1 \TRANSMAP e'_1$ and $e_2 \TRANSMAP e'_2$. By definition of $\HBIO$ induced by I/O events on input and output sets, any values $v_1\in\mathit{IO}(e_1)$ and $v_2\in\mathit{IO}(e_2)$ are such that $v_1 \HBIO v_2$. Since $\TRANS$ is a valid transformation, events $e'_1$ and $e'_2$ also have to be ordered such that $v_1 \HBIO v_2$, hence $e'_1 \HBIO e'_2$.
\end{proof}

Finally, one may lift the notion of transformation preservation to a program instruction, collecting events associated all or a subset of the executions of this instruction.

\begin{definition}[Transformation-preserved instruction]
  \label{def:tpi}
  Given a program $P$, $i_{\mathit{tp}}$ is a \emph{transformation-preserved instruction} of $P$ if all valid program transformations are guaranteed to preserve its associated events, for all inputs.
  Formally,
  \begin{multline*}
  \forall i_{\mathit{tp}} \in P,\ i_{\mathit{tp}} \text{ is a transformation-preserved instruction if and only if } \\
  \forall \TRANS \in \mathcal{T}(P), \forall I \in \mathit{Inputs}, \forall e_{\mathit{tp}} \in \EXEC{P}{I}, i_{\mathit{tp}}=\INST(e_{\mathit{tp}}),\ \exists e'_{\mathit{tp}} \in \EXEC{\TRANS(P)}{I}, e_{\mathit{tp}} \TRANSMAP e'_{\mathit{tp}}
  \end{multline*}

  And $i_{\mathit{tp}}$ is \emph{transformation-preserved conditionally on the preservation of an instruction} $i_c$ if for all executions of $P$, the preservation of some event $e_c$ associated with the execution of $i_c$ implies the preservation of any event $e_{\mathit{tp}}$ associated with the execution of $i_{\mathit{tp}}$.
  Formally,
  \begin{multline*}
  \forall i_{\mathit{tp}} \in P,\ i_{\mathit{tp}} \text{ is conditionally transformation-preserved on } i_c \text{ if and only if } \\
  \forall \TRANS \in \mathcal{T}(P), \forall I \in \mathit{Inputs}, \forall e_{\mathit{tp}} \in \EXEC{P}{I}, i_{\mathit{tp}}=\INST(e_{\mathit{tp}}),
  \exists e_c \in \EXEC{P}{I}, e'_c \in \EXEC{\TRANS(P)}{I},\\
  e_c \TRANSMAP e'_c \land \INST(e_c)=i_c \implies
  \exists e'_{\mathit{tp}} \in \EXEC{\TRANS(P)}{I}, e_{\mathit{tp}} \TRANSMAP e'_{\mathit{tp}}
  \end{multline*}
\end{definition}

We will use these notions to validate the preservation of security protections, either for all possible executions, or conditionally on the execution of a given secure expression/function.

\subsection{Observation Semantics}

Let us now consider the last expression in our Mini IR syntax. \ic{snapshot} expressions introduce a specific mechanism to observe values along the execution of the program. Vu et al.\ \cite{Vu20} defined a
\emph{partial observation trace} as a sequence of sets of $(\mathrm{variable}, \mathrm{value})$ and $(\mathrm{address}, \mathrm{value})$ pairs. To increase the reach of compiler optimizations while
preserving the user's ability to attach logical properties to specific values and instructions, we extend the observation semantics to \emph{partially ordered partial states} defined by the execution of
instructions involving \ic{snapshot} expressions.

\begin{definition}[Partial state]
Any event involving a \ic{snapshot} expression define a \emph{partial observation state} of the operational semantics, or \emph{partial state} for short. These partial states are modeled by the following \emph{observation function}:
\[
\mathit{Obs}: \mathit{Events} \to \mathit{States}
\]
extracting from an event $e$ a \emph{partial state} $(\mathit{ObsV}, \mathit{ObsC}, \mathit{ObsR}, \mathit{ObsM}, \pi)$ such that $\pi$ is the program point of the instruction associated with $e$ and $\mathit{ObsV} \subseteq \mathit{V}$, $\mathit{ObsC} \subseteq \mathit{C}$, $\mathit{ObsR} \subseteq \mathit{R}$, $\mathit{ObsM} \subseteq \mathit{M}$ are the $(\mathit{name}, \mathit{value})$ pairs observed by all arguments of \ic{snapshot} expressions involved in event $e$.
\end{definition}

In addition, an individual instruction involving a \ic{snapshot} expression returns all its arguments in addition to capturing these arguments' $(\mathit{name}, \mathit{value})$ pairs into a partial state.

Let us now define \emph{observation events} associated with the execution of instructions involving \ic{snapshot} expressions.

\begin{definition}[Observation event]
  \label{def:oe}
  We call \emph{observation event} any event associated with the execution of an instruction involving a \ic{snapshot} expression.
\end{definition}

The following definitions introduce the observation-ordering relation as a
precise tool in the hand of the programmer to define an ordering between
observation events; this relation on observation events is derived from def-use,
reference-based and in-memory data-flow relations, and control dependences.

\begin{definition}[Dependence relation]
  We define relations $\HBDU$, $\HBRF$ and $\HBCD$ as partial orders on def-use pairs, in-reference/in-memory data flow and control dependences, respectively.  Formally, let $\mathit{def}(v, i)$ and $\mathit{use}(v, i)$ be the predicates evaluating to true if and only if instruction $i$ defines variable $v$ and instruction $i$ uses variable $v$, respectively, and let $\mathit{postdom}$ denote the post-domination binary predicate:
\begin{eqnarray*}
  e_1 \HBDU[1] e_2 && \text{if and only if}\quad
  \mathit{def}(v, \INST(e_1)) \land \mathit{use}(v, \INST(e_2))\\
  e_1 \HBRF[1] e_2 && \text{if and only if}\quad
  \big( \INST(e_1) = (\mathit{ref}\texttt{ <- }v) \lor \INST(e_1) = (\mathic{mem[}\mathit{addr}\texttt{] <- }v) \big)\\
  & \land\ & \big( \INST(e_2) = (\mathit{var}\texttt{ = }\mathit{ref}) \lor \INST(e_2) = (\mathit{var}\texttt{ = }\mathic{mem[}\mathit{addr}\texttt{]}) \big)\\
  & \land\ & \nexists e_s,\ E = \ldots e_1 \ldots e_s \ldots e_2 \ldots,\\
  && \quad\big( \INST(e_s) = (\mathit{ref}\texttt{ <- }v') \lor \INST(e_s) = (\mathic{mem[}\mathit{addr}\texttt{] <- }v') \big)\\
  e_1 \HBCD[1] e_2 && \text{if and only if}\quad
  \exists e_s,\ E = \ldots e_1 \ldots e_s \ldots e_2 \ldots,\\
  && \quad\mathit{postdom}(\INST(e_2), \INST(e_s)) \land \neg\mathit{postdom}(\INST(e_2), \INST(e_1))
\end{eqnarray*}
and
\[
  \HBDU\ = \big(\HBDU[1]\big)^* \qquad
  \HBRF\ = \big(\HBRF[1]\big)^* \qquad
  \HBCD\ = \big(\HBCD[1]\big)^*
\]
  
  The dependence relation, denoted by $\HBDEP$, is defined as the union of the
  def-use, reference-based and in-memory data-flow, and control dependence relations:
  \[
  \HBDEP[1]\ =\ \HBDU[1] \cup \HBRF[1] \cup \HBCD[1]
  \qquad \textrm{and} \qquad
  \HBDEP\ = \big(\HBDEP[1]\big)^*
  \]
\end{definition}

\begin{definition}[Observe-from relation]
\label{def:of}
Given an execution $E$, observation events induce a relation called \emph{observe from} and denoted by $\HBOF$, mapping a definition to an observation event $e_{\mathit{obs}}$:
\begin{align*}
  \forall \ldots e_1 \ldots e_{\mathit{obs}} \ldots \in E,\ i = \INST(e_1),\ {} & \big( \mathit{def}(v,i)\ \land\ \INST(e_{\mathit{obs}}) = (\mathit{var, ...}\texttt{ = }\text{\mathic{snapshot}}(v, ...)) \big)\\
                                                                             {} & \implies e_1 \HBOF e_{\mathit{obs}}
\end{align*}
\end{definition}

\begin{definition}[Observation ordering relation]
Any pair of distinct events associated with the execution of instructions involving \ic{snapshot} expressions related through a dependence relation are ordered by a so-called
\emph{observation ordering} relation denoted by $\HBOO$. Formally, given an execution $E$ of $P$, $\HBOO$ is the restriction of $\HBDEP$ to observation events:

%% Son: This does not express the implicit
%% ordering of snapshot expressions of memory and reference. For instance, consider this event sequence:\\
%% (1) ref $<-$ v;\\
%% (2) var = ref;\\
%% (3) snapshot(var);\\
%% (4) ref $<-$ v';\\
%% (5) var' = ref;\\
%% (6) snapshot(var');\\
%% Ideally, we want (3) oo (6) since we observe values from the same reference. Currently, we have (1) dep (2), (2) of (3), (4) dep (5), (5) of (6), but (3) has no relation with (4), thus (3) cannot oo (6). I guess a read followed by a write to the same reference/memory address should also be part of dep relation.  Or maybe we do not want to handle implicit ordering here? Or am I missing something? AC: the DEP relation is important, HBOO is a restrciction of DEP on observation events, the example you propose misses a DEP relation between (3) and (6).
\begin{align*}
  \forall \ldots e_1 \ldots e_2 \ldots \in E,\\
  \big( \INST(e_1) = (\mathit{var_1, ...}\texttt{ = }\text{\mathic{snapshot}}(v_1, ...))\ {} & \land\ \INST(e_2) = (\mathit{var_2, ...}\texttt{ = }\text{\mathic{snapshot}}(v_2, ...))\ \land\ e_1 \HBDEP e_2 \big)\\
                                                                                       {} & \implies e_1 \HBOO e_2
\end{align*}
\end{definition}

We chose to only include data flow relations (through
SSA values, references or memory) and control dependences into $\HBOO$. This is a trade-off between
providing more means to the programmer to constrain program transformations to
enforce observation ordering, and freedom left to the compiler in presence of
such observations.
Data-flow paths between \ic{snapshot} expressions enable the expression of arbitrary partial orders of observation events, and they are easily under the control of programmers---if necessary by inserting dummy or token values as we will see in the next section---hence they appear to be expressive enough for our purpose.
Note that control dependences are not directly useful at capturing partial ordering (that would not be otherwise expressible using data dependences), and it is sufficient to \emph{not} make them dependent on
the result of \ic{snapshot} expressions to avoid having to unduly constrain the ordering of observations (e.g., forbidding legitimate hoisting of loop-invariant expressions). On the other hand, control
dependences are important to model the effect of program transformations converting data dependences into control dependences: for example, boolean logic may be converted into control flow, yielding a single
static truth value for some boolean variables occurring in a dependence chain linking two observations.
Conversely, adding more relations into $\HBOO$ such as non-data-flow write-after-write (memory-based) dependences would not enhance the ability to represent more partial orders while severely restricting the compiler's ability to reorder loop iterations or hoist observations from loops.

\subsection{Happens-Before Relation}

We now define a partial order on both I/O and observation events, capturing not only the I/O semantics of the program but also its associated observations.

\begin{definition}[Happens-before relation]
  \label{def:hb}
  For a given program execution $E$, one may define a partial order $\HB$ over pairs of events called a \emph{happens-before relation}. It has to be a sub-order of the total order of events in $E$.
\end{definition}

\begin{definition}[Preservation of happens-before]
  \label{def:hbp}
  Given a valid program transformation $\TRANS$, for any input $I \in \mathit{Inputs}$, $P$ produces an execution $E = \EXEC{P}{I}$, and the transformed program $P' = \TRANS(P)$ produces an execution $E' = \EXEC{P'}{I}$. $\TRANS$ is said to preserve the happens-before relation if any events in happens-before relation in $P$ have their counterparts through $\TRANSMAP$ in happens-before relation in $P'$. Formally,
  \[
    \forall e_i, e_j \in E,\ \forall e'_i, e'_j \in E',\ e_i \HB\ e_j\ \land\ e_i \TRANSMAP e'_i\ \land\ e_j \TRANSMAP e'_j\ \implies\ e_i' \HB\ e_j'
  \]
\end{definition}

The preservation of the happens-before relation is a property that has to be proven in general. Depending on how sparse the $\TRANSMAP$ and $\HB$ relations are, it may be more or less difficult to enforce and establish. In the following, we use the following happens-before relation:
\[
\HB\ = \big( \HBIO \cup \HBOF \cup \HBOO \big)^*
\]
Thanks to Lemma~\ref{lem:pioeo} one will only have to prove the preservation of the $\HBOF$ and $\HBOO$ components of $\HB$ in the following. On the contrary, unlike I/O instructions, instructions
involving \ic{snapshot} are \emph{not} preserved by valid program transformations in general.

Let us now provide two important definitions to reason about the preservation of observations.

\begin{definition}[Observation-preserving transformation]
  \label{def:opt}
  Given a program $P$, a transformation $\TRANS$ that applies to $P$ is \emph{observation-preserving} if it produces a program $P'=\TRANS(P)$ such that the four following conditions hold:
  \begin{enumerate}
  \item[(i)] it is a valid transformation (see Definition~\ref{def:valid});
  \item[(ii)] it preserves the existence of observation events:
      \begin{multline*}
        \forall I \in \mathit{Inputs}, \forall e \in \EXEC{P}{I},\ \INST(e) =
          (\mathit{var, ...}\texttt{ = }\mathic{snapshot}(v, ...))\\
        \implies \exists e' \in \EXEC{P'}{I},\ \INST(e') = (\mathit{var',
          ...}\texttt{ = }\mathic{snapshot}(v', ...))\ \land\ e \TRANSMAP e'
      \end{multline*}
  \item[(iii)] it preserves all happens-before relations:
  \[
    \forall I, \forall e_1, e_2 \in \EXEC{P}{I},\ e_1 \HB e_2 \implies \exists e'_1, e'_2 \in \EXEC{P'}{I}, e'_1 \HB e'_2
  \]
  \item[(iv)] it preserves the observed values:
      \begin{align*}
        \forall I \in \mathit{Inputs}, \forall {} & e \in \EXEC{P}{I},\ \INST(e)
          = (\mathit{var, ...}\texttt{ = }\mathic{snapshot}(v, ...)),\\
        {} & e' \in \EXEC{P'}{I},\ \INST(e') = (\mathit{var', ...}\texttt{ =
          }\mathic{snapshot}(v', ...))\ \land\ e \TRANSMAP e'\\
        {} & \implies \mathit{Obs}(e) = \mathit{Obs}(e')
      \end{align*}
  \end{enumerate}

  Given a program $P$, a transformation $\TRANS$ that applies to $P$ is \emph{observation-preserving conditionally on} instruction $i_c$ in $P$ if it produces a program $P'=\TRANS(P)$ such that the conditions (i), (iii) and (iv) above hold, and also:
  \begin{enumerate}
  \item[(ii$_c$)] it preserves the existence of observation events conditionally on the preservation of $i_c$:
      \begin{align*}
        \forall I \in \mathit{Inputs}, {} & \forall e \in \EXEC{P}{I},\ \INST(e)
          = (\mathit{var, ...}\texttt{ = }\mathic{snapshot}(v, ...))\ \land\\
        {} & \exists e_c \in \EXEC{P}{I}, e'_c \in \EXEC{P'}{I}, \INST(e_c) =
          i_c, e_c \TRANSMAP e'_c\\
        {} & \implies \exists e' \in \EXEC{P'}{I},\ e \TRANSMAP e'
      \end{align*}
  \end{enumerate}  
\end{definition}

Let us now define a notion of observation that is preserved over all possible valid transformations.

\begin{definition}[Protected observation]
  \label{def:protected}
  An observation in a program $P$ is \emph{protected} if and only if all valid transformations that apply to $P$ are observation-preserving.

  An observation is \emph{protected conditionally on instruction} $i_c$ if and only if all valid transformations are observation-preserving conditionally on $i_c$.
\end{definition}

As expected, preservation (resp.\ protection) imply conditional preservation (resp.\ protection) on all instructions.

Note that the composition of two valid transformations yields a valid transformation, according to Definition~\ref{def:valid}. As a consequence, Definition~\ref{def:protected} covers compositions of valid transformations along a compilation pass pipeline.

In the next section, we will provide a constructive method to implement programs with \emph{protected observations} complying with Definition~\ref{def:protected}. This will allow us to prove the preservation
of $\HB$ on a class of programs with carefully defined protections of \ic{snapshot} expressions, as a partial fulfillment of the requirements for a valid program transformation to be observation-preserving.

\section{Value Preservation Mechanisms}
\label{sec:proposed}

As noted earlier, valid transformations do not preserve the happens-before relation in general. This section introduces a mechanism to achieve this, involving a minor extension of the Mini IR with expressions that are defined to be opaque to any program analysis.

\subsection{Opaque Expressions}
\label{sec:opaque}

\begin{figure}[h!tb]
  \begin{cadre}
    \footnotesize
    \newenvironment{syntax}{\begin{tabular}{@{}rrll@{}}\spacefalse}{\end{tabular}}

\begin{syntax}
  \nonterm{extended-expr}{} \is{} \nonterm{expr}{} & \textrm{} \alt{} \addspace\lstinline[keywordstyle={}]$opaque$\indextt{opaque}\spacetrue \nonterm{region}{} & \textrm{ atomic opaque region: make I/O, side-effects and definitions} \newl{}                                  & \textrm{ visible atomically outside the region and vice versa;} \newl{}                                  & \textrm{ the compiler sees statically unknown yet functionally} \newl{}                                  & \textrm{ deterministic values} \alt{} \addspace\lstinline[keywordstyle={}]$yield($\indextt{yield}\spacetrue \nonterm{var}{}\repetstar{} \addspace\lstinline[keywordstyle={}]$)$\spacetrue & \textrm{ return from atomic opaque region} \sep{}\end{syntax}

%%% Local Variables:
%%% mode: latex
%%% TeX-master: "main"
%%% End:

    \normalsize
  \end{cadre}
  \caption{Extension of Mini IR to implement event and happens-before preservation.}
\label{fig:ext}
\end{figure}

To implement the preservation of observation events and the associated happens-before relation, we extend Mini IR with an \emph{opaque expression} syntax. The \ic{opaque} keyword introduces a region of control flow, as shown in \figRef{fig:ext}. The opaque expression syntax gets its name from the ``opacity'' of its enclosed region w.r.t.\ program analyses and transformations. An instruction defining a value from an opaque expression is called an \emph{opaque instruction}.

When executing \ic{opaque}, the associated event gathers the definitions and effects of all instructions in its enclosed region, \emph{atomically and in isolation}. We assume the enclosed region is a \emph{terminating} sequence
of instructions. It proceeds with ``internal'' state transitions without defining events and without exposing intermediate states. When reaching a \ic{yield} instruction, the program state serves as the
resulting state of the atomic event while also defining all values listed in the \ic{yield} instruction. Formally, executing an opaque instruction enclosing a region
\texttt{\{$i_1$; \ldots; $i_n$\}} on a state $\sigma$ yields the event $e = \sigma \LEADSTO{a} \sigma'$ where $\sigma \LEADSTO{i_1} \cdots \LEADSTO{i_n} \sigma'$. The last instruction $i_n$ must be a \ic{yield} instruction.
We authorize arbitrary (terminating) control flow in these regions, including conditional memory access and I/O. As a result, the memory and I/O effects triggered by an opaque instruction are input-dependent.
Since regions inside opaque expressions always terminate, $\sigma'$ always exists. This definition guarantees both atomicity and isolation, since states and transitions associated with individual
instructions $\{i_k\}_{1\ \le\ k\ \le n}$ are not modeled in the operational semantics. Finally, the semantics of nested opaque expressions is defined inductively from the inside out.

The compiler is very limited in what analyses it may perform on opaque expressions:
\begin{itemize}
\item gathering the uses of an opaque expression;
\item deciding whether an opaque expression has read or write side-effects;
\item deciding whether an opaque expression performs I/O;
\item deciding whether two opaque expressions are identical up to variable renaming.
\end{itemize}
Yet the compiler is not allowed to determine the precise side-effects (references, memory addresses) in an opaque expression, and it may not attempt to establish a correlation between its uses (resp.\ loads from references or memory) and the values it defines (resp.\ stores to references or memory).

On the other hand, a valid transformation $\TRANS$ may also delete or duplicate an opaque instruction, and even synthesize
completely new ones. This may sound too powerful, as without additional care $\TRANS$ may break the opacity and expose intermediate states in an opaque expression. We will see in the following that the
opacity property itself prevents this from happening, thus maintaining opacity, atomicity and isolation of opaque expressions across transformations.

Since opaque expressions can nest multiple instructions (and even nested regions), we introduce a notation to denote the sequence of instructions executed atomically within an event.
Let $\INSTLIST(e)$ denote the list of instructions associated with event $e$; it is a single-element list for all events, except for those associated with opaque
instructions where it is the sequence of instructions executing within the region for this particular instance $e$ of the opaque instruction. We will write $i \in \INSTLIST(e)$ to denote that an instruction $i$ is associated with event $e$.

In the rest of the paper, we will use revised and extended versions of Definitions~\ref{def:io_ordering}--\ref{def:opt} operating on \emph{sets of instructions} in $\INSTLIST(e)$ rather than a specific instruction $\INST(e)$. All equalities of the form $i = \INST(e)$ in these equations should be rewritten into $i \in \INSTLIST(e)$. For convenience, we will also consider all I/O expressions as being opaque; this is consistent with the traditional assumptions about compilers not being able to analyze across system calls.

Informally, opaque expressions have two important consequences on value and event preservation across program transformations: (1) if a valid program transformation preserves an event using a value defined by an opaque instruction or stored by an instruction from its associated region, then it must also preserve the event associated with the opaque instruction, and (2) a valid program transformation has to preserve any value used in the opaque expression, as proceeding with downstream computation would otherwise involve some form of unauthorized guessing of the opaque expression's behavior.
Formally, let the predicate $\mathit{Opaque}(e)$ denote that event $e$ is
associated with the execution of an opaque instruction; we restrict the effects of program transformations in presence of opaque expressions as follows:

\begin{definition}[Opaque expression preservation]
  \label{def:opaque}
  Given a program $P$, input $I$, and valid transformation $\TRANS$,
  \begin{multline}
    \forall \ldots e_1 \ldots e_2 \ldots \in \EXEC{P}{I},\ \mathit{Opaque}(e_1),\ e_1 \HBDEP[1] e_2,\\
    \exists e'_2 \in \EXEC{\TRANS{(P)}}{I},\
    e_2 \TRANSMAP e'_2\
    \implies \exists e'_1 \in \EXEC{\TRANS{(P)}}{I},
    \ e_1 \TRANSMAP e'_1
    \ \land\ \mathit{Opaque}(e'_1)
    \ \land\ e'_1 \HBDEP e'_2
    \label{eqn:opaqueuse}
  \end{multline}
  \begin{multline}
    \forall \ldots \sigma_e \ldots \in \EXEC{P}{I},\
    \forall \ldots \sigma'_{e'} \!\ldots \in \EXEC{\TRANS{(P)}}{I},\
    \mathit{Opaque}(e),\ e \TRANSMAP e',\\
    \mathit{use}(v, \INST(e))\ \land\ (v, \mathit{val}) \in \sigma_e
    \implies
    \exists i' \in \INSTLIST(e'),
    \mathit{use}(v', i')\ \land\ (v', \mathit{val}) \in \sigma'_{e'}
    \label{eqn:opaqueval}
  \end{multline}
\end{definition}

This restriction is taken as a definition, capturing formally the intuitive expectations about what the compiler has to enforce in the presence of opaque expressions.

Notice the transitive dependence relation $e_1' \HBDEP e_2'$ in the transformed program (rather than $e_1' \HBDEP[1] e_2'$): the immediate dependence may be transformed into a series of instructions (e.g., spilling a value to the stack).

Let us highlight a subtle point in this definition: $i'$ is an instruction
belonging to the transformed opaque expression's region, not the opaque
instruction itself. Indeed, variable $v'$ may not be a free variable in $\INST(e')$, it may be bound to the opaque expression's internal region. For example, it is always correct to transform \ic{opaque \{ some_use_of(v) \}} into the sequence \ic{t1 = not v; opaque \{ t2 = not t1; some_use_of(t2) \}}.
This does not involve any analysis of the opaque expression's semantics---which is explicitly disallowed. While \ic{t2} remains part of the program state and retains the value \ic{v} had in the original program, it is not exposed as a variable captured by the opaque expression.

In the following, we will only use \ic{snapshot} within the region of an opaque expression. As a result, \ic{snapshot} expressions will inherit all properties of opaque expressions, including the conditions for their preservation (\ref{eqn:opaqueuse}) and the preservation of observed values (\ref{eqn:opaqueval}).

\subsection{Opaque Chains}

Let us now build dependence chains involving opaque instructions. These will be called \emph{opaque chains} and serve two purposes: (1) establishing a transformation-preserved $\HBOO$ relation, and (2) linking observations to downstream I/O events to preserve the former through program transformations. We first need additional definitions and notations.

\begin{definition}[Opaque value set]
  \label{def:opaquevalue}
  Given an execution $E=\EXEC{P}{i}$, consider a chain of dependent events $e_1 \HBDEP[1] \ldots \HBDEP[1] e_n$ with $n\ge 2$, an opaque instruction/event $i_j=\INST(e_j)$ on the chain defining a variable $\mathit{var}_j$ and an instruction/event $i_k=\INST(e_k)$ on the chain, with $1\le j<k\le n$ and $\forall j<l<k, \neg\mathit{Opaque}(e_l)$. Let $\sigma_j$ be the program state $e_j$ transitions into.

  $\mathit{OV_j}$ denotes the set of all opaque values that $\mathit{var}_j$ may take according to its data type: for example, an opaque expression yielding a value of boolean type will have $\mathit{OV_j}=\{\textrm{true}, \textrm{false}\}$.

  We also lift this definition to the set of values used by a downstream expression across a chain of dependent instructions. Consider an execution $E_{\mathit{alt}}$ continuing after $e_j$ on program state $\sigma_j\{\mathit{var}_j\mathbin{\mapsto}\mathit{alt}\}$.\footnote{The substitution syntax $\sigma_j\{\mathit{var}_j\mathbin{\mapsto}\mathit{alt}\}$ denotes the set $\sigma_j\setminus(\mathit{var}_j,\ldots)\cup(\mathit{var}_j,\mathit{alt})$.} Consider an event $e_{k_{\mathit{alt}}}$ such that $e_j \HBDEP e_{k_{\mathit{alt}}}$ and $E_{\mathit{alt}} = \ldots e_j\ldots e_{k_{\mathit{alt}}} \ldots$. $\mathit{value}_{j,k}$ denotes the function mapping every value $\mathit{alt}$ in $\mathit{OV_j}$ to a value defined as follows:
  \begin{itemize}
  \item $\mathit{value}_{j,k}(\mathit{alt})$ is the value used or read by $i_{\mathit{alt}}=\INST(e_{k_{\mathit{alt}}})$ along the $e_j\ldots e_{k_{\mathit{alt}}}$ sub-chain if $i_k$ and $i_{\mathit{alt}}$ are identical expressions up to variable renaming, and $\forall e\neq e_{k_{\mathit{alt}}}\ \textrm{s.t.}\ e_j \HBDEP e \HBDEP e_{k_{\mathit{alt}}},\ \neg\mathit{Opaque}(e)$;
  \item $\mathit{value}_{j,k}(\mathit{alt})$ is the special value $\bot$ otherwise.
  \end{itemize}
  We define the opaque value set $\mathit{OV_{j,k}}$ as $\mathit{value}_{j,k}(\mathit{OV_j})$.
\end{definition}

Intuitively, the definition of $\mathit{value}_{j,k}$ allows to reason on the cardinality of the opaque value set $\mathit{OV_{j,k}}$: whether this set is a singleton or not will tell whether the dependent instruction/event $i_k=\INST(e_k)$ is truly sensitive on the opaque value of $i_j=\INST(e_j)$. A non-singleton set tells that evaluating the opaque instruction $i_j$ cannot be avoided by a valid transformation.
Let us paraphrase this definition to expose the intuitions behind it. When substituting the value $\mathit{alt}$ for the opaque expression $i_j$, the $\mathit{value}_{j,k}$ function yields the value used or read by $i_k$ if the execution path of the altered execution still reaches $i_k$. If the execution path is altered when considering the $\mathit{alt}$ value and encounters an identical instruction/event $i_{\mathit{alt}}=\INST(e_{k_{\mathit{alt}}})$ before encountering a dependent opaque instruction, it yields the value used or read by $i_{\mathit{alt}}$ (this accounts for program transformations capable of combining two identical instructions, more on this later). And if the altered execution reaches an opaque instruction before reaching an identical instruction, $\mathit{value}_{j,k}$ yields the ``undefined'' value $\bot$.

\begin{definition}[Opaque chain]
  \label{def:opaquechain}
  Given an execution $E=\EXEC{P}{i}$, consider a chain of dependent events $e_1 \HBDEP[1] \ldots \HBDEP[1] e_n$;
  $e_1 \HBDEP[1] \ldots \HBDEP[1] e_n$ is an \emph{opaque chain} linking $e_1$ to $e_n$ if and only if
  \begin{itemize}
  \item[(i)] $i_1=\INST(e_1)$ and $i_n=\INST(e_n)$ are opaque instructions;
  \item[(ii)] for $2 \le k \le n$, an \emph{opaque} instruction $i_k = \INST(e_k)$, and $1 \le j < k$ the immediately preceding opaque instruction on the chain, $\mathrm{Card}(\mathit{OV}_{j,k}) \ge 2$.
  \end{itemize}
  We note $e_1 \OPAQUETO e_n$ such an opaque chain.

  The intuition behind the (ii) condition is the following. Given an
  instruction/event $i_k=\INST(e_k)$, for the immediately upstream opaque
  instruction/event $i_j=\INST(e_j)$ on the chain, we consider all values it may
  define according to its opaque result type. Either $i_k$ is control-dependent
  on $i_j$ and there exists an alternate execution from $e_j$ bypassing $i_k$
  (or an identical expression up to variable renaming), or $i_k$ is
  data-dependent on $i_j$ and the set of values $i_k$ may use or read is not a
  singleton, or both.
\end{definition}

Opaque chains take the form of an alternating sequence of opaque instructions and sub-chains of regular instructions, starting with an opaque instruction and ending with an opaque instruction (remember I/O expressions are considered opaque).
The control- and data-dependence restrictions in Case~(ii) serve as ``information-carrying'' guarantees: the compiler does not have enough information about the possible paths dependent on an opaque value or on the processing of opaque values to break an opaque chain into distinct dependence chains.

Let us consider a few examples of dependence chains that are not opaque chains.
First, considering data dependences only, shifting an opaque \ic{uint32_t} by 32 bits to the right would allow the compiler to reason about the resulting zero value, transforming the downstream opaque expression into one applied to the constant zero, hence breaking the chain. The cardinality requirement on data dependences rules such transformations out.
As a more complex example, let us illustrate the restrictions on the multiplicity of paths leading and not leading to $i_k$ or an identical opaque expression. Consider the following program
\begin{lstlisting}
bb_entry:
  c = opaque {
    yield(42);
  };
  br c, bb_true;
bb_false:
  io(desc, 0);
  br true, bb_join;
bb_true:
  io(desc, 0);
bb_join:
\end{lstlisting}
forming a dependence chain from the definition of \ic{c} to the I/O instruction---through a control dependence. And consider its transformation into
\begin{lstlisting}
bb_entry:
  c = opaque {
    yield(42);
  };
  io(desc, 0);
\end{lstlisting}
As a result of \emph{instruction combining}, there is no dependence anymore in the transformed program. The dependence chain in the original program is not an opaque one due to the identical values (constant $0$) read by the I/O instruction on both paths leading to an identical instruction.

On the contrary, the following example illustrates the conversion of control into data dependences and vice-versa, preserving opaque chains in the process: consider the programs $P_{\textrm{data}}$
\begin{lstlisting}
bb_entry:
  c = opaque {
    u = snapshot(boolean_input);
    yield(u);
  };
  br c, bb_true;
bb_false:
  v = 0;
  br true, bb_join;
bb_true:
  v = 42;
bb_join:
  opaque {
    snapshot(v);
  };
\end{lstlisting}
and $P_{\textrm{control}}$
\begin{lstlisting}
bb_entry:
  c = opaque {
    u = snapshot(boolean_input);
    yield(u);
  };
  br c, bb_true;
bb_false:
  opaque {
    snapshot(0)
  };
  br true, bb_join;
bb_true:
  opaque {
    snapshot(42);
  };
bb_join:
\end{lstlisting}
Both form opaque chains from the definition of \ic{c} to the observation of \ic{v}, \ic{42}, \ic{0}; the transformations from $P_{\textrm{data}}$ to $P_{\textrm{control}}$ and vice-versa are both valid, preserving observations (a data dependence is converted into a control dependence, specializing values into constants, and vice-versa for the reverse transformation). Whether it is the multiple values of \ic{v} or the alternative path from the definition of \ic{c} to a consuming \ic{snapshot}, it is impossible for the compiler to break the dependence.

Beyond opaque instructions, important classes of instructions always belong to an opaque chain:
\begin{itemize}
\item all instructions that only propagate existing values within or across name-value domains; these include dereference, assignment, load, store, \ic{br} on branch arguments (not on the branch condition), \ic{call} and \ic{return} instructions, and instructions involving \ic{snapshot} or \ic{io} expressions;
\item the same applies to the traditional C unary operators \texttt{-}, \texttt{!}, \texttt{\char`~};
\item any binary operator (resp.\ function call) where one or more operands (resp.\ arguments) are opaque, and where opaque operands (resp.\ arguments) are not correlated (feeding multiple times the same opaque value or dependent expressions may degenerate into a singleton value set, such as the subtraction of an opaque value with itself);
\item any binary operator (resp.\ function call) where the operand (resp.\ arguments) type or the value of the other operand (resp.\ other arguments) makes the operation bijective; e.g., \texttt{+} on unsigned integers, \texttt{*} with the constant 1, etc.
\end{itemize}
More instructions may belong to an opaque chain provided specific constraints hold on its inputs: e.g., left-shifting by 1 an \ic{unsigned int} value if the compiler cannot prove it is always greater than or equal to \ic{UINT\_MAX/2}, or dividing a value that the compiler cannot statically analyze to be less than the divisor; in both cases the compiler is forced to consider that the image of the instruction on all possible inputs is not a singleton.

We may now generalize Equation~(\ref{eqn:opaqueuse}) to opaque chains.
If a valid program transformation preserves an event at the tail of an opaque chain, then it must also preserve the event associated with the head of the opaque chain. Formally:

\begin{lemma}[Chained opacity]
  \label{lem:chainedopaque}
  Given a program $P$, input $I$, and valid transformation $\TRANS$,
  \begin{multline}
    \forall \ldots e_1 \ldots e_n \ldots \in \EXEC{P}{I},\ e_1 \OPAQUETO e_n,
    \land\ \exists e'_n \in \EXEC{\TRANS(P)}{I}, e_n \TRANSMAP e'_n\\
    \implies \exists e'_1 \in \EXEC{\TRANS{(P)}}{I},
    \ e_1 \TRANSMAP e'_1
    \ \land\ \mathit{Opaque}(e'_1)
    \ \land\ e'_1 \HBDEP e'_n
    \label{eqn:chainedopaque}
  \end{multline}
\end{lemma}

\begin{proof}
  The $\OPAQUETO$ relation implies $e_1$ is opaque.
  If $n=1$, the result stems from the application of Equation~(\ref{eqn:opaqueuse}) to $e_1$.
  Consider the case of $n>1$. The only instructions in our Mini IR capable of
  implementing a non-constant function
  %% \TODO{i.e.\ modifying the cardinality? AC: no, non-constant map from opaque value to uses across an opaque chain.}
  (the cardinality restriction in the
  definition of $\OPAQUETO$) are the definition, load and store instructions, as
  well as conditional branches. These are exactly the instructions building up
  $\HBDEP$. As a result, the fact the value used by $e_n$ is sensitive to the
  value defined by $e_1$ implies the existence of a slice of events in $P$
  implementing this non-constant function, and spawning backward from $e_n$; in
  addition $e_1$ belongs to this backward slice. From
  Equation~\ref{eqn:opaqueval}, $e_n$ being opaque, the mapping of $e_n$ to
  $e'_n$ implies that the same value computed by a non-constant function is used
  by $e'_n$. Once again, this non-constant function is computed by a slice in
  $\TRANS(P)$ spawning backward from $e'_n$. The backward slice in $P$ includes
  one or more instructions depending on the result of $e_1$, including $e_2$.
  This is also the case in the transformed program: $\EXEC{\TRANS(P)}{I}$ holds
  an event $e'$ that uses the opaque value defined by $e_1$. Either $\INST(e')$
  is an opaque instruction lumping together $\INST(e_1)$ with additional instructions including the use of the opaque value defined by $e_1$
  %% \TODO{What does ``extending'' exactly mean? AC: Reworded, I hope it is clear.}
  and defining $e'_1=e'$ and $e_1 \TRANSMAP e'_1$ concludes the proof, or one may define $e_2 \TRANSMAP e'$ and apply Equation~(\ref{eqn:opaqueuse}) again.

  %% \TODO{$e'_1 \HBDEP e'_n$, $\mathit{Opaque}(e'_1)$, $\mathit{Opaque}(e'_n)$
  %%     (because $e_n \TRANSMAP e'_n$), all we have to do to prove $e'_1 \OPAQUETO
  %%     e'_n$ is the cardinality requirement. Indeed, without details about other
  %%     transformed instructions of the chain, we cannot conclude about the
  %%     cardinality. AC: Indeed. We are not trying to prove that the transformed dependence chain is an opaque chain. It would be great to do so, but we failed to do that until now, and have a way to circumvent this limitation. Are you suggesting we add a comment after the proof to highlight this fact?}
\end{proof}

Let us now introduce an important lemma on the preservation of opaque expressions.

\begin{lemma}[Preservation of opaque chains]
  \label{lem:opaque_preservation}
  Given a program $P$ and input $I$, if $e_1$ is linked through an opaque chain to a transformation-preserved event $e_n$, then $e_1$ is transformation-preserved, and for any transformation $\TRANS$ mapping $e_1$ to $e'_1$ and $e_n$ to $e'_n$, there is a chain of dependent instructions linking $e'_1$ to $e'_n$. Formally,
  \begin{multline*}
  e_1 \OPAQUETO e_n\ \land e_n\in\mathit{TP}(P,I) \implies
  \forall \TRANS\in\mathcal{T}(P),\ \exists e'_1, e'_n \in \EXEC{\TRANS(P)}{I},\ e_1 \TRANSMAP e'_1\ \land\ e'_1 \HBDEP e'_n.
  \end{multline*}
\end{lemma}

\begin{proof}
  Consider an opaque chain $e_1 \OPAQUETO e_n$.
  The existence of $e'_n$ derives from the hypothesis that $e_n$ is transformation-preserved.
  
  The case of $n=1$ also stems immediately from the same hypothesis that $e_n$
  is transformation-preserved.

  Consider a chain of $n>1$ events $e_1 \OPAQUETO e_n$. Let $i_n = \INST(e_n)$,
  and consider the last opaque instruction/event $i_k = \INST(e_k)$ before
  $e_n$. Applying Equation~(\ref{eqn:chainedopaque}) to $e_k$
  %% \TODO{This also implies that a sub-chain of an opaque chain is also an opaque chain. Need to be
  %% explicitly stated in the opaque def? AC: it depends what you call a sub-chain; but yes with a proper definition it would be true. Still, I don't see any value in defining a notion of (opaque) sub-chain of an opaque chain. The definition of opaque chains clearly applies to any intermediate opaque $e_k$.}
  implies the
  existence of $e'_k$ such that $e_k \TRANSMAP e'_k$ and $e'_k \HBDEP e'_n$.
  Induction on the length of the chain proves the preservation of $e_1$ through
  valid transformations for all chain lengths, and that the transformed events
  form a dependence chain $e'_1 \HBDEP e'_n$.
\end{proof}

Notice that a typical case of transformation-preserved $e_n$ is an I/O event, but the lemma is not limited to opaque chains terminating in a consuming I/O event. Notice also that events associated with non-opaque expressions along the chain are not necessarily preserved.

Unfortunately, despite our efforts to make the definition as robust to transformations as possible, we have not been able to prove that an opaque chain transforms into an opaque chain in general. Additional hypotheses on opaque expressions are likely to be needed to prove this, and we conjecture these would not require modifying our definition of opaque chains, but we do not have a definite answer at this point. Fortunately, we do not need to prove such a strong preservation result: a weaker-yet-sufficient compositionality result can be used as a work-around.

\begin{lemma}[Transitive preservation of opaque chains]
  \label{lem:transitive_opaque_preservation}
  \begin{multline*}
    e_1 \OPAQUETO e_n\ \land\ e_n\in\mathit{TP}(P,I) \\
    \implies \forall\TRANS\in\mathcal{T}(P),\forall \TRANS'\in\mathcal{T}(\TRANS(P)),\ \exists e''_1, e''_n \in \EXEC{\TRANS'(\TRANS(P))}{I},\ e_1 \TRANSMAP[\TRANS'\circ\TRANS] e''_1\ \land\ e''_1 \HBDEP e''_n
  \end{multline*}
\end{lemma}

\begin{proof}
  The proof of transitive transformation preservation stems from the observation that $\TRANS \circ \TRANS'$ is a valid transformation and the application of Lemma~\ref{lem:opaque_preservation}.
\end{proof}

We will use Lemma~\ref{lem:transitive_opaque_preservation} to reason about the preservation of opaque events with a dependent transformation-preserved event (e.g.\ an I/O instruction) throughout the compilation flow.

\subsection{Syntactic Sugar}

Let us now introduce a simple pattern implementing a ``tokenizing'' opaque expression.
%% AC: I use pattern to introduce a programming patterm, not as a replacement for macro. I also try not to use macro in formal definitions/lemmas/theorems, but use its name only.

\begin{lstlisting}
macro token(v1, ..., vk) { // pure, opaque unit-type value,
                           // not associated with any resource;
                           // the compiler sees a statically unknown
                           // yet functionally deterministic value
bb_macro:
  w = opaque {
    use(v1, ..., vk);
    yield(unit_value);
  };
  return(w);
}
\end{lstlisting}
where the variadic \ic{use} function uses all its arguments and returns no value.

It is pure, functionally deterministic, and opaque to the compiler. Semantically, the \ic{token} pattern returns a value of the unit type---called a token---irrespectively of the number of arguments. Yet it is not known to the compiler what values a token can take, and in particular the compiler is not told it is a singleton set. As a result one may use \ic{token} as an opaque expression, and to use the token type in dependent instructions forming an opaque chain. Unlike more general opaque expressions, token values do not need hardware resources to store live token values when emitting assembly code: we refer to the unique value \ic{unit_value} of the unit-type to implement ``resource-less'' opaque chains of tokens.

It may sound paradoxical to return the same \ic{unit_value} in multiple tokens, yet this is not visible to the compiler since it occurs in an opaque expression; the compiler must assume these are different, unpredictable values.

Furthermore, since the unit data type does not carry an informative value, we define \ic{snapshot} to ignore its token arguments, i.e., not to embed them into a partial state.

Finally, we introduce a \ic{tailio} pattern implementing a special case of the \ic{io} instruction with no argument. It is used in opaque chains to prevent them from being eliminated.
\begin{lstlisting}
macro tailio() {
bb_macro:
  desc = tagged_unit_unordered_set_descriptor;
  io(desc);
  return();
}
\end{lstlisting}
The \ic{tailio} pattern uses a dedicated descriptor of an unordered set of (tagged) unit-valued
outputs. Since it embeds an I/O instruction, it is preserved by transformations, yet unlike
the more general form of I/O it does not incur any ordering constraint.

\subsection{Robust Partial State Observation}

Based on the extended syntax, we may now define a series of observation patterns, defining partial states preserved over program transformations. They are ordered by increasing degrees of freedom for the scheduler and optimizations in general.

\paragraph{Monolithic opaque expression}

The simplest form of observation preserved over program transformations.

\begin{lstlisting}
macro observe_monolithic(v1, ..., vk) {
bb_macro:
  t2 = opaque {
    w1, ..., wk = snapshot(v1, ..., vk);
    t1 = token(w1, ..., wk);
    tailio();
    yield(t1);
  };
  return(t2);
}
\end{lstlisting}

\begin{lemma}[Preservation of monolithic observation events]
  \label{lem:monolithic}
  If (1) all instructions involving \ic{snapshot} expressions of a program $P$ are introduced via \ic{observe_monolithic}, and if (2) the happens-before relation on observations in $P$ is enforced through an opaque chain, then observations in $P$ are protected according to Definition~\ref{def:protected}.
\end{lemma}

\begin{proof}
  Consider any valid transformation $\tau$. We need to prove that $\tau$
  preserves observations, according to \defRef{def:opt}.

  The opaque expression's atomic region expands a nested \ic{tailio}
  pattern. As a result, \defRef{def:pioe} guarantees that all events
  associated with the execution of such an opaque expression are
  transformation-preserved. This proves conditions~(i) and~(ii) in
  \defRef{def:opt}.

  Any valid transformation must preserve the value of (token) \ic{t1} from
  the definition of opaque expressions.\footnote{Remember the compiler assumes tokens can
  have multiple values, even if these do not consume any resources.} The
  opaque expression's region holds a backward slice linking atomically the
  (token) value \texttt{t1} to the arguments $\texttt{v1}, \ldots, \texttt{vk}$
  to be observed.  By equation (2) in Definition~\ref{def:opaque}, any valid
  transformation must preserve the values of $\texttt{v1}, \ldots, \texttt{vk}$
  along with the definition events producing these values. This proves
  condition~(iv) in Definition~\ref{def:opt}.

  Let us now prove condition~(iii)---the preservation of $\HB$. The preservation
  of $\HBOF$ follows the same reasoning as the proof of condition~(iv); yet the
  preservation of $\HBOO$ involves the traversal of opaque chains. Consider a
  pair of events $e_{\mathit{obs}_1}$, $e_{\mathit{obs}_2}$, each of which is
  associated with the opaque expression expanded from \ic{observe_monolithic},
  such that $e_{\mathit{obs}_1} \HBOO e_{\mathit{obs}_2}$, and let
  $e'_{\mathit{obs}_1}$, $e'_{\mathit{obs}_2}$ be their counterparts in
  $P'=\TRANS(P)$; these transformed events must exist as they involve I/O
  effects. Since $e_{\mathit{obs}_1} \OPAQUETO e_{\mathit{obs}_2}$,
  Lemma~\ref{lem:opaque_preservation} applies to all \ic{opaque} events on the
  opaque chain, guaranteeing their mapping to events in $P'$ through $\TRANS$,
  and that these events form a dependence chain. This proves
  $e'_{\mathit{obs}_1} \HBOO e'_{\mathit{obs}_2}$.  
\end{proof}

Notice that some arguments of \ic{observe_monolithic} associated with $e_{\mathit{obs}_2}$ may be converted by $\TRANS$ into constants, which are defined at the initial event $e_0$. Still, at least one argument will remain the target of the opaque chain linking $e_{\mathit{obs}_1}$ to $e_{\mathit{obs}_2}$.

\paragraph{Decoupled I/O region}

Cutting out a specific I/O section allows to decouple event preservation from event ordering. Rather than inserting an I/O instruction in every observation, it is sufficient to consider a smaller set, or even a single I/O instruction, and to chain this set backwards to the observations whose events they are meant to preserve.
\begin{lstlisting}
macro observe_decoupled(v1, ..., vk) {
bb_macro:
  u1 = opaque {
    w1, ..., wk = snapshot(v1, ..., vk);
    t1 = token(w1, ..., wk);
    yield(t1);
  };
  return(u1);
}

macro observe_tailio(u1, ..., uk) {
bb_macro:
  v = opaque {
    u = token(u1, ..., uk);
    tailio();
    yield(u);
  };
  return(v);
}
\end{lstlisting}

The programmer may use \ic{observe_decoupled} to organize observations, setting partial ordering constraints among them using the pattern's resulting token. These will be much less intrusive to compiler
transformations than \ic{observe_monolithic} and its embedded I/O instruction.

We can now adapt the Lemma~\ref{lem:monolithic} to the decoupled preservation pattern.

\begin{lemma}[Preservation of decoupled observation events]
  \label{lem:decoupled}
  If (1) all \ic{snapshot} instructions of a program $P$ are introduced via
  \ic{observe_decoupled}, if (2) the happens-before relation on observations in
  $P$ is enforced through opaque chains whose tails are I/O events introduced
  via \ic{observe_tailio}, then observations in $P$ are protected according to
  Definition~\ref{def:protected}.
\end{lemma}

\begin{proof}
  Similarly to Lemma~\ref{lem:monolithic}, we need to
  prove that any valid transformation $\tau$ preserves decoupled observation
  event, according to Definition~\ref{def:opt}.

  Since \ic{observe_tailio} includes a \ic{tailio} instruction,
  \defRef{def:pioe} guarantees that events associated
  with instances of the opaque expression holding the \ic{tailio} instruction
  are transformation-preserved. Lemma~\ref{lem:opaque_preservation}
  applied to opaque chains linking \ic{observe_decoupled} opaque expressions to
  a subsequent \ic{observe_tailio} proves conditions~(i) and~(ii) in
  Definition~\ref{def:opt}.
  
  The proof of condition~(iii)---the preservation of $\HB$---in Lemma~\ref{lem:monolithic} applies to \ic{observe_decoupled} as it does not refer to the \ic{tailio} instruction.

  The proof of condition~(iv) is also identical to the proof of Lemma~\ref{lem:monolithic}.
\end{proof}

\ic{tailio} or \ic{observe_tailio} typically occur at the tail of an
opaque chain, and a single \ic{tailio}/\ic{observe_tailio} may close
multiple opaque chains.

\subsection{I/O-Based Schemes for Partial State Observation}

Let us now use our extended syntax to implement the functional property
preservation mechanism described in \cite{Vu20}. The authors define the functional property
preservation as the preservation of (1) values occurring in the property predicates and (2) the program observation point at which the property is evaluated. As such, the original mechanism is divided in two
parts: Vu et al.'s algorithm inserts an artificial definition for every definition reaching a functional property, which corresponds to inserting the following \ic{artificial_def} pattern, and an observation
for the functional property itself, implemented with the following \ic{observe_cc} pattern.

\begin{lstlisting}
macro artificial_def_cc(v) {
bb_macro:
  u = opaque {
    desc = ordered_set_descriptor;
    io(desc);
    yield(v);
  };
  return(u);
}

macro observe_cc(u1, ..., uk) {
bb_macro:
  opaque {
    w1, ..., wk = snapshot(u1, ..., uk);
    desc = ordered_set_descriptor;
    io(desc);
    yield();
  };
}
\end{lstlisting}

Using the first pattern, every value referred in the property's predicate is protected by the compiler using an opacification mechanism. In the original approach, these opaque values next replace the original ones
in the subsequent code. Finally, the second pattern snapshots the opaque values at the observation point. Notice that unlike our lightweight approach, all opaque expressions used for property preservation
contain an I/O instruction inducing a total order. This total ordering
constraint is cumbersome, and not always wanted by the programmer but rather a downside of the approach.

\subsection{Observing Address-Value Pairs}

When observing values in memory, several applications require observing not only the value but also the memory address that holds this value. For example, this is important when assessing the proper erasure
of a buffer in memory, to avoid leaking sensitive data. The following pattern provides such a functionality, when associated with an \ic{observe_tailio} pattern as in the decoupled or fine-grained schemes above.

\begin{lstlisting}
macro observe_pair(a) {
bb_macro:
  u = opaque {
    v = mem[a];
    b, w = snapshot(a, v);
    t = token(b, w);
    yield(t);
  };
  return(u);
\end{lstlisting}

The observation of both address \ic{a} and the value stored at this location \ic{a} occurs atomically w.r.t.\ other observations since they belong to the same partial state. One may use \ic{observe_pair} to check the value at a specific memory address, as required by the abovementioned memory erasure example.

Of course, one may also define a version of this pattern for the monolithic scheme, and versions with a variable number of address-value pairs.

\subsection{Value-Preserving Observation-Opacification Pattern}

We will see in Sections~\ref{sec:implementation} and ~\ref{sec:security} that a common pattern to protect secure implementations consists in opacifying specific values, without modifying the data or control flow. In particular, rather than building an opaque chain of tokens---like \ic{observe_decoupled} does---it is natural to chain observations using the original values but hiding them from potentially harmful transformations.

\begin{lstlisting}
macro observe_and_opacify(v1, ..., vk) {
bb_macro:
  u1 = opaque {
    w1, ..., wk = snapshot(v1, ..., vk);
    yield(w1);
  }
  return(u1);
}
\end{lstlisting}

The \ic{observe_and_opacify} pattern implements the identity function on its first argument. Its other arguments can be used to express data dependence relations. In addition, all arguments are observed. The compiler
sees the result of \ic{observe_and_opacify} as a statically unknown yet functionally deterministic value. Any observation introduced via \ic{observe_and_opacify} is called \textit{opacification}.

Notice that \ic{observe_and_opacify} only differs from \ic{observe_decoupled} in the returned value (the first argument rather than a token). As a consequence, the proof of Lemma~\ref{lem:decoupled} applies directly to the following value-preserving observation and opacification result.

\begin{lemma}[Preservation of observation-opacification events]
  \label{lem:observe_and_opacify}
  If (1) all \ic{snapshot} instructions of a program $P$ are introduced via
  \ic{observe_and_opacify}, if (2) the happens-before relation on opacifications in
  $P$ is enforced through opaque chains whose tails are I/O events introduced
  via \ic{observe_tailio}, then observations in $P$ are protected according to
  Definition~\ref{def:protected}.
\end{lemma}

\subsection{Protecting Observations and Logical Properties}

Collect all the results of this section, the following result provides
a methodology for the observation and opacification of any value(s) or address-value pair(s), the enforcement of any partial ordering among these observations, and the preservation of both observations and their partial ordering.

\begin{theorem}[Observation protection]
  \label{thm:protection}
  Let $P$ be a program implementing observations through \ic{observe_monolithic}
  using opaque chains to enforce a programmer-specified $\HB$ order. Then all
  observations in $P$ are protected according to Definition~\ref{def:protected}.

  Let $P$ be a program implementing observations through a combination of \ic{observe_monolithic}, \ic{observe_decoupled} and \ic{observe_and_opacify} on opaque chains enforcing a programmer-specified $\HB$ order, and such that any chain involving \ic{observe_decoupled} and \ic{observe_and_opacify} leads to a downstream transformation-preserved instruction (such as a trailing \ic{tailio}). Then all observations in $P$ are protected according to Definition~\ref{def:protected}.

  Let $P$ be a program implementing observations through a combination of \ic{observe_monolithic}, \ic{observe_decoupled} and \ic{observe_and_opacify} on opaque chains enforcing a programmer-specified $\HB$ order, and such that any chain involving \ic{observe_decoupled} and \ic{observe_and_opacify} leads to a downstream instruction transformation-preserved conditionally on some instruction in a set $I_c$. Then all observations in $P$ are protected conditionally on instructions in $I_c$ according to Definition~\ref{def:protected}.
\end{theorem}

The theorem above relates to observation preservation. More generally, to support the evaluation of a logical property such as an ACSL formula \cite{ACSL}, one may need to observe an arbitrary partial state
with a variable number of $(\mathit{name}, \mathit{value})$ pairs. The monolithic, decoupled and observation-opacification patterns above achieve this, by collecting all name-value domains occurring in the logical property and atomically observing their name- and address-value pairs.

\section{Putting it to Work}
\label{sec:implementation}

Let us now describe the implementation of our approach across multiple levels of program representation through an optimizing compilation flow, ranging from source code,
compiler IR, down to binary code. We focus on C programs, representative of secure embedded applications, and we build our framework on the LLVM compilation infrastructure. It is composed of three main
phases. The first one is the front-end which translates high level programming languages into the LLVM IR. As our source programs are written in C, we naturally choose Clang, the LLVM front-end for C, C++
and Objective-C. Then, in a second phase, the LLVM IR generated by Clang is processed by the source-and-target-independent middle-end optimizers. In the last phase, the back-end (or the code generator)
lowers the LLVM IR produced by the optimizers to a machine-specific representation called LLVM MIR, which supports both SSA and register-allocated non-SSA forms; the LLVM MIR is subject to a few late machine
code optimizations, before finally being converted to assembly code.

%% As a side note, it would require little work to extend this choice in order to handle source programs written in other embedded programming languages, such as C++ or Ada.

%% AC: At some point in the future it would be interesting to extend the approach to support (language-level) exceptions.

\subsection{Value Preservation in Source Code}

We introduce language extensions to Clang to support observation (\texttt{snapshot}) and opacification (\texttt{opaque}) expressions. We define three builtins \ic{\_\_builtin\_opacify}, \ic{\_\_builtin\_observe\_mem} and \ic{\_\_builtin\_io} corresponding, respectively, to the \ic{observe\_opacify}, \ic{observe\_pair} and \ic{observe\_tailio} macros defined in the previous section.
\begin{itemize}
\item \ic{\_\_builtin\_opacify} is a variadic function implementing
  value opacification. It returns the same scalar value as its first
  argument, but made opaque to the compiler. This opaque value may
  then replace the original one in subsequent code. All other
  arguments are optional and represent additional data dependence
  relations to implement opaque chains constraining program
  transformations. The builtin function also observes (snapshots) all
  its arguments: this provides a means
  %% \TODO{a meanS ? }. AC: yes, this is a nound.
  to validate the opacification mechanism by tracing observed values down to the generated machine code.
%     The second argument is a boolean indicating whether or not other arguments are observed by the compiler. 
% Its returned value could replace the first argument
% value in the subsequent code, or potentially be consumed to
% establish an explicit data dependence relation, or simply to prevent the observation from being removed by optimizations.

  \item \ic{\_\_builtin\_observe\_mem} implements an address-value pair observation. It reads a value from its pointer argument and observes (snapshots) this value together with the pointer itself. It returns a token to implement downstream opaque chains.

  \item \ic{\_\_builtin\_io} is a variadic function implementing an I/O effect: the arguments serve to extend opaque chains linking upstream observations to the I/O effect. The function returns a token to implement downstream opaque chains.
\end{itemize}
These language extensions enable the programmer to define additional constraints when transforming the program, in the form of data dependences or ordering relations. As unit type is not natively defined in C, we currently use integer-typed variables to represent tokens produced by \ic{\_\_builtin\_observe\_mem} or \ic{\_\_builtin\_io} and only used by our builtins. There are two main cases: when builtins are removed and this eliminates all uses of a token variable, this variable is eliminated as well and does not incur any resource overhead; when the token variable remains live due to escaping values (in function call or return, or in memory), this variable, will incur low resource usage in the generated machine code, most likely a single stack slot for the whole function and no register usage beyond the token definition on a RISC ISA. A better solution would be to implement a fully expressive token type in LLVM (the current one is limited---it may not be used in $\phi$ nodes---and has a different, provenance-tracking purpose).

\subsection{Value Preservation in LLVM}

Let us now describe the transformation of our language extensions to two different compiler intermediate representations: the IR on which the optimizers operate and the MIR which represents the final code to
be emitted by the compiler.

\subsubsection{Value Preservation in LLVM IR}

LLVM IR supports intrinsic (a.k.a.\ builtin) functions with compiler-specific semantics. Intrinsics require the compiler to follow additional rules while transforming the program. These rules are communicated to the compiler via the \emph{function attributes} which specify the intrinsic function's behavior w.r.t.\ the program mutable state (memory, control registers, etc.). Intrinsics provide an extension mechanism without having to change all of the transformations in the optimizers.

To implement our preservation mechanism, we introduce three intrinsics to the LLVM IR:
\begin{itemize}
\item \ic{llvm.opacify} has the same semantics as \ic{\_\_builtin\_opacify}.
  % It opaquely produces a new SSA value from its first argument, so that compiler optimizations cannot reason about the relation between these two values, while other optional arguments produced by some preceding instructions ensure that the latter are scheduled before the intrinsic. All values are also observed by the compiler.
  It is pure, does not access memory and has no I/O or other side-effects: it is valid to optimize away \ic{llvm.opacify} if the opaque value is not used in subsequent code.

\item \ic{llvm.observe.mem} has the same semantics as \ic{\_\_builtin\_observe\_mem}.
  % It reads from the memory that its pointer-typed argument points to then this value is observed. It returns a token represented as an SSA value.
  % It is meant to be used by some succeeding instructions, either to prevent transformations such as CSE or DCE from optimizing away the intrinsic, or to ensure that the intrinsic is scheduled before these instructions.
  Unlike \ic{llvm.opacify}, \ic{llvm.observe.mem}'s attributes let it read argument-pointed memory. Other than such reads it has no side effects: it is valid to optimize away \ic{llvm.observe.mem} if the output token is not used in subsequent code. Being able to access argument-pointed memory is actually an optimizing implementation of the \ic{observe\_pair} macro presented in \secRef{sec:proposed}: this avoids having to generate instructions loading from these memory locations.
% In order to observe values from the memory value domain, \ic{llvm.opacify} is defined as an instruction that may read memory pointed to by its pointer-typed argument(s), and that has no I/O or other
% side-effects. This means that optimizations can remove the intrinsic if there is no use of the output token in subsequent code. Being able to read argument-pointed memory guarantees \ic{llvm.opacify} to be
% scheduled after reaching stores to its pointer-typed argument(s), without the compiler having to generate memory load(s) from these memory locations.

\item \ic{llvm.io} has the same semantics as \ic{\_\_builtin\_io}.
  % It takes a variable number of SSA values as arguments and is defined as an side-effecting instruction, so that it cannot be removed by optimizations, thus makes its argument values always live. It returns a token represented as an SSA value.
\end{itemize}

We modified Clang to map \ic{\_\_builtin\_opacify}, \ic{\_\_builtin\_observe\_mem} and \ic{\_\_builtin\_io} to \ic{llvm.opacify}, \ic{llvm.observe.mem} and \ic{llvm.io} respectively, when generating LLVM IR
from C code; observations are represented in LLVM IR as metadata attached to the corresponding intrinsic. LLVM IR metadata is indeed designed to convey additional information to optimizers and code
generators \cite{llvm_md}, and we defined a new type of metadata carrying information about observed values: (1) the observation builtin source-level identifier such as source line and column number, (2) the program
point at which the observation takes place, and (3) the location holding the observed value at this point. We modify a few utility functions commonly used by different optimization passes such as
\ic{replaceAllUsesWith()} and \ic{combineMetadata()} to update (e.g.\ when
combining two intrinsics) and maintain (e.g.\ when duplicating the intrinsic)
the metadata attached to the intrinsic throughout the optimization pipeline. An
alternative would have been to embed metadata directly into the
intrinsic (by passing a metadata operand to the intrinsic), avoiding
to modify these utility functions to update and maintain observation metadata. But this would
prevent optimizations from combining intrinsics when values are equal but metadata differs (e.g.\ line and column numbers).

\subsubsection{Value Preservation in LLVM MIR}

To preserve values throughout code generation we also need to implement our mechanism in the MIR. We achieve this by lowering the intrinsics \ic{llvm.opacify}, \ic{llvm.observe.mem} and \ic{llvm.io} respectively into \ic{OPACIFY}, \ic{OBSERVE\_MEM} and \ic{IO} pseudo-instructions, with the same semantics and behaviors w.r.t.\ memory accesses and side effects. Pseudo-instructions are MIR instructions that do not have machine encoding information and must be expanded, at the latest, before code emission. Nevertheless, our mechanism should not interfere with the emitted machine code; the pseudo-instructions introduced are thus not expanded but completely removed during code emission. To guarantee the correct functional behavior of the program when removing the pseudo-instructions, the \ic{OPACIFY} pseudo-instructions uses the same register as its first operand to hold the opaque value.

The preservation of observation metadata is more challenging: LLVM does not currently support attaching metadata to MIR instructions, we thus have to transform IR metadata into an operand of \ic{OPACIFY} and \ic{OBSERVE\_MEM} pseudo-instructions. This may require modifications to passes in the code generator to maintain and update the metadata while not preventing them from optimizing the program: indeed, as discussed in the alternative implementation above, this approach is potentially preventing some optimizations from combining pseudo-instructions with the same arguments but different observation metadata; fortunately we did not have to do so since we did not find any such missed optimizations on our benchmark suite and on the different backends considered.

At the final stage of the code generator, observation metadata is also emitted into machine code in the debug section. This allows to communicate information about the observed values to binary code utilities carrying out the validation of observation and opacification mechanisms (the debugger, binary code verifiers, etc.). To represent this information in machine code, we extend the DWARF format \cite{dwarf}, which provides an easily extensible description of the executable program. Yet unlike the more conventional approach relying on debug information generated by the compiler itself \cite{Vu20}, we maintain and update the information of observed values ourselves, only using DWARF for its standard encoding of the data, since it is already supported by most binary code utilities.

\section{Preserving Security Protections}
\label{sec:security}

\lstset{
  language=newc
}

It has been shown that there is a correctness-security gap in compilation, which
arises when compiler optimizations preserves the functional semantics but
violates a security guarantee made by source program \cite{silva}. As a
consequence, security engineers have been fighting with optimizing compilers for
years by devising and introducing complex programming tricks to the source code,
though yet found a reliable way to obtain secure binary code \cite{simon}. In
this section, we demonstrate, on different examples, how our mechanisms can be
used to preserve security protections through an optimizing compilation downto
the generated binary.

\subsection{Sensitive Memory Data Erasure}
\label{security:erasure}

First, let us consider our motivating example described in \secRef{sec:example}.
The security protection consists in erasing a secret buffer allocated on the
stack after usage; however, most compilers will spot
that the buffer is not accessible after the function returns, removing
the call to \ic{memset()} as part of ``dead store elimination''.

To preserve the erasure, we insert an opaque artificial read of values stored in
the buffer, after the call to \ic{memset()}. We then use the value produced by
the opacification in an I/O effecting operation to guarantee that it does not
get removed, as shown in \lstRef{lst:erasure-blt}. We validate the approach on mbedTLS's RSA encryption and decryption \cite{mbedtls}, called \ic{erasure-rsa-enc} and \ic{erasure-rsa-dec} in the following. A
short opaque chain links \ic{__builtin_observe_mem()} to the final I/O builtin. The former is also an observation that enables the validation of the security property (i.e., effective erasure).

\begin{figure}[h!tb]
  \begin{lstlisting}[xleftmargin=0.35cm, caption={Erasing a buffer with observation.},
                     captionpos=b, label={lst:erasure-blt}]
void process_sensitive(void) {
  uint8_t secret[32];
  ...
  memset(secret, 0, sizeof(secret));
  __builtin_io(__builtin_observe_mem(secret));
}
  \end{lstlisting}
\end{figure}

%To verify the existence of all sensitive buffer erasures, we perform a
%comparison, for the same program input, of two \textit{observation traces}:
%\begin{enumerate}
%  \item Reference trace: we execute/simulate the reference program at \ic{O0}.
%  The reference program is the original program (without builtins) with
%  \texttt{printf} inserted to emulate the observation events and output values
%  stored in the sensitive buffer after each erasure. As optimizations are
%  disabled at \ic{O0}, the reference trace is guaranteed to contain all correct
%  observed values at every observation point.
%  \item Optimized trace: we simulate the program with builtins inserted at
%  different optimization levels (\ic{O1/2/3/s/z}): using a debugger/simulator,
%  we break the simulation at each observation point reported in the DWARF
%  section and output the observed values for the corresponding observation event.
%\end{enumerate}
%    
%Since the secret memory data erasures are placed at the end of different
%functions, we have the same ordered sequence of observation events in both
%observation traces.

\subsection{Computation Order in Masking Operation}
\label{security:comp_order}

Something as simple as respecting the computation order, as explicitly written
in the source code, of associative operations, may be difficult to achieve when
compiling with optimizations enabled. Indeed, as long as the generated program
produces matching observable behaviors w.r.t.\ the C standard, compilers
have perfectly the right to reorder associative operations, even with proper
parenthesizing, and doing so independently of the optimization level.

Now, it may sound like no big deal to reorder associative operations, because
after all, this is what associativity really means. Nonetheless, this can be
problematic when it comes to using associative operations such as \ic{xor} for
masking against side-channel attacks. In fact, masking operation is a commonly
used countermeasure to protect block cipher implementations against side-channel
attacks \cite{masking}. Consider the code excerpt of a masking scheme shown in
\lstRef{lst:comp-order-orig}. The secret \ic{k} is first masked with \ic{m}
(line 1), then is remasked with \ic{mpt} (line 3). Note how the programmer has
intentionally put the parentheses to express the fact that \ic{k} has to be
remasked with \ic{mpt} before removing the old mask \ic{m}. Nevertheless, it has
been reported that the statement in line~3 has been compiled as
\ic{k}$~\hat{}$~\ic{(mpt}~$\hat{}$~\ic{m)}, which altogether defeats the
countermeasure \cite{eldib}.

\begin{figure}[h!tb]
  \begin{lstlisting}[xleftmargin=0.35cm, caption={Masking example.},
                     captionpos=b, label={lst:comp-order-orig}]
k ^= m;
...
k = (k ^ mpt) ^ m;
  \end{lstlisting}
\end{figure}

In order to preserve the correct order in the masking operation, we propose a
solution based on opacification, as shown in \lstRef{lst:comp-order-blt}. We
first linearize the compound expression to three-address form by explicitly
declaring a temporary variable \ic{tmp} to hold the result of the remasking
operation (line~3), which is next used for unmasking (line~ 5). Obviously, this
alone would only guarantee the correct masking operation if no optimizations
enabled. To prevent compiler optimizations from removing \ic{tmp} and reodering
the masking operation, we opacify the result of the remasking operation, making
it unknown to the compiler (line~4). Furthermore, we assign the opaque value to \ic{tmp} to make sure that subsequent code refers to this value instead of the original one, thus guaranteeing the use of the
correct value in the removal of the old mask \ic{m}. This forms an opaque chain from the definition of \ic{tmp} to the definition of \ic{k}. There is no need for a terminal I/O builtin since we already know that the computation of \ic{k} is transformation-preserved, \ic{k} being the value of interest in downstream computation. Notice also that the cardinality constraint on values in opaque chains is trivially satisfied by the bijectivity of the exclusive or operator. The opaque chain enforces the ordering constraint that the opacified value of \ic{tmp} will be observed after the first masking operation and before the second one.

We validate our approach on a masked implementation of Advanced Encryption Standard (AES) \cite{aes}, named \ic{mask-aes} in the following.
In the following, we will also consider a self-written application called \ic{mask-rotate}, which contains a loop of masking operations with the same security property as \ic{mask-aes}, together with I/O instructions; the goal is to evaluate the performance overhead of our lightweight implementation relying on pure intrinsics without side effects.

\begin{figure}[h!tb]
  \begin{lstlisting}[xleftmargin=0.35cm, caption={Secure masking using opacification.},
                     captionpos=b, label={lst:comp-order-blt}]
k ^= m;
...
uint8_t tmp = k ^ mpt;
tmp = __builtin_opacify(tmp);
k = tmp ^ m;
  \end{lstlisting}
\end{figure}

%We verified that (1) \ic{tmp} did indeed hold
%the expected result of the remasking operation, and (2) \ic{tmp} was actually
%used in the unmasking operation. The former is done by making the opacification
%event output the opacified value \ic{tmp}, thus allow us to obtain an optimized
%observation trace and compare it to a reference trace, as explained in
%\ssecRef{val:erasure}. As for the latter, verifying values in traces is not
%enough to show that \ic{tmp} was used in the subsequent computation. We thus
%have to perform a data flow analysis which verifies the effective use of
%\ic{tmp} in the unmasking operation. The analysis is carried out on the late
%LLVM MIR, right before code emission, as a compiler back-end pass.

\subsection{Step Counter Incrementation}
\label{security:sci}

Fault attacks are a growing threat for secure devices such as smart
cards. Such attacks can alter the system's correct behavior via
physical injection means~\cite{Yuce2018}. For example, it has been
shown that fault attacks can induce unexpected jumps to any location
in the program \cite{moro, Berthome2012}. One source-level scheme to enhance the
resilience against such fault attacks \cite{lalande} is shown in the code
excerpt from \lstRef{lst:sci}. The protection consists in defining a step
counter at each control construct (line 2), and stepping the counter of the
immediately enclosing control construct after every C statement of the original
source (lines 4 and 6). Counters are then checked against their expected values
at the exit of the enclosing control construct (lines 8 and 9), calling a
handler when it fails (line 10). We refer to this technique as Step Counter
Incrementation (SCI), which may be seen as a very fine-grained form of Control
Flow Integrity (CFI)
\cite{Abadi:2005:CI:1102120.1102165,Burow:2017:CIP:3058791.3054924}.

\begin{figure}[h!tb]
  \begin{lstlisting}[xleftmargin=0.35cm, caption={SCI protection.},
                     captionpos=b, label={lst:sci}]
...
unsigned short cnt_if = 0;
if (cond) {
  cnt_if++;
  a = b + c + 1;
  cnt_if++;
}
if (!((cnt_if == 2 && cond) ||
      (cnt_if == 0 && !cond)))
  fault_handler();
  \end{lstlisting}
\end{figure}

However, as fault attacks are not modeled in compilers, optimizations are free
to transform the program even when it does not preserve the security
countermeasures inserted by the programmer. Indeed, counter checks are
removed---their conditions are trivially true in a ``fault-free'' semantics of the
program. Counter incrementations might hence be removed, or grouped into a
single block of code. As a result, practitioners making use of this
source-level hardening scheme have to disable compiler optimizations. Instead,
we make use of our opacification mechanism to preserve the SCI protection, as shown
in \lstRef{lst:sci-blt}. In fact, preserving the SCI protection boils down to
(1) protecting counter incrementations and checks and (2) guaranteeing the
proper interleaving of functional and countermeasure statements. The former
can be achieved by opacifying counters at each of their incrementations (lines 4
and 6), and at checks against expected constant values (lines 8 and 9), so that
checks can no longer be deduced, while counter values can no longer be
constant-propagated and must be incremented instead. As for the latter, we need
to create additional data dependences between values defined by the functional
code and counter values: we insert an artificial use of the counter value at
each functional value definition and inversely, an artificial use of the last
functional value defined at each counter incrementation.
To achieve this, we opacify non-constant operands used in
definitions of functional or counter values and express these artificial uses as
token parameters of the opacification operator (lines~5 and~6). This creates an opaque chain linking every counter incrementation to the next counter use, and then again to the next incrementation until the terminating fault handler, while interleaving original program statements in the chain through the bundling of both counter and original variables in opacification builtins. Notice that the opaque chain includes a control dependence when linking with the fault handler. We validated this approach on two well-known smart-card benchmarks: PIN authentication
\cite{fissc} and AES encryption \cite{aes256}, called \ic{sci-pin} and
\ic{sci-aes} in the following.

\begin{figure}[h!tb]
  \begin{lstlisting}[xleftmargin=0.35cm, caption={Secure SCI protection using opacification and data dependences.},
                     captionpos=b, label={lst:sci-blt}]
...
unsigned short cnt_if = 0;
if (cond) {
  cnt_if = __builtin_opacify(cnt_if) + 1;
  a = __builtin_opacify(b, cnt_if) + __builtin_opacify(c, cnt_if) + 1;
  cnt_if = __builtin_opacify(cnt_if, a) + 1;
}
if (!((__builtin_opacify(cnt_if) == 2 && cond) ||
      (__builtin_opacify(cnt_if) == 0 && !cond)))
  fault_handler();
  \end{lstlisting}
\end{figure}

\subsection{Source-Level Loop Protection}
\label{security:loop}

Recent research has shown that loops are particularly sensitive to fault
attacks. Indeed, faulted iteration counter in AES cipher could lead to retrieval of
secret key \cite{dehbaoui}, while fault injections in the core loop of memory copy
operation during embedded system's boot stages may allow an attacker to control
the target's execution flow which eventually will lead to arbitrary code
execution on the target \cite{timmers}. Other work highlighted the need to
protect the iteration counter of the PIN code verification on smart cards
\cite{fissc}.

\begin{figure}[h!tb]
  \begin{lstlisting}[xleftmargin=0.35cm, caption={Original \ic{memcmp()} implementation.},
                     captionpos=b, label={lst:loop}]
int memcmp(char *a1, char *a2, unsigned n) {
  for (unsigned i = 0; i < n; ++i) {
    if (a1[i] != a2[i]) {
      return -1;
    }
  }
  return 0;
}
  \end{lstlisting}
\end{figure}

To enforce the correct iteration counter of sensitive loops, a compile-time loop
hardening scheme has been recently proposed and implemented in LLVM \cite{proy}.
It operates on the LLVM IR and is based on the duplication of loop termination
conditions and of the computations involved in the evaluation of such
conditions. However, such redundant operations do not impact the program
observable semantics and are ideal candidates to be optimized away by downstream
optimizations \cite{hillebold}. The authors has originally investigated and
analyzed different compilation passes in order to select a relevant position for
the loop hardening pass in the compilation flow, so that the countermeasure is
preserved in the executable binary. We argue that the investigation of the
positioning of compile-time countermeasure pass can be facilitated, if not
dismissed. We implement the loop hardening scheme at source level and leverage
our opacification mechanism to preserve the redundancy-based protection through
the whole optimizing compilation flow. Consider an implementation of
\ic{memcmp()} function, shown in \lstRef{lst:loop}.

\begin{figure}[h!tb]
  \begin{lstlisting}[xleftmargin=0.35cm, caption={Securing \ic{memcmp()} loop.},
                     captionpos=b, label={lst:loop_protected}]
int memcmp(char *a1, char *a2, unsigned n) {
  unsigned i, i_dup, n_dup = n;
  for (i = 0, i_dup = 0; i < n; ++i, ++i_dup) {
    if (i_dup >= n)
      fault_handler();
    if (a1[i] != a2[i]) {
      if (a1[i_dup] == a2[i_dup])
        fault_handler();
      if (n_dup != n)
        fault_handler();
      return -1;
    }
  }
  if (i_dup < n)
    fault_handler();
  if (n_dup != n)
    fault_handler();
  return 0;
}
  \end{lstlisting}
\end{figure}

\lstRef{lst:loop_protected} demonstrates our approach on the core loop of the
above \ic{memcmp()} function. We duplicate the loop counter
\ic{i} (line~3) and loop-independent variables being used in the loop body
(\ic{n} in this case, line~2). Furthermore, we insert redundant computations of
the exit condition at every iteration of the loop (line~4), as well as at
the loop exit (lines~7 and~14). We also verify that the values of the duplicated
loop-independent variables at every loop exit are correct w.r.t.\ the
values of their original counterparts (lines~9 and~16).

\begin{figure}[h!tb]
  \begin{lstlisting}[xleftmargin=0.35cm, caption={Secure \ic{memcmp()} loop using opacification.},
                     captionpos=b, label={lst:loop_blt}]
int memcmp(char *a1, char *a2, unsigned n) {
  unsigned i, i_dup, n_dup = __builtin_opacify(n);
  for (i = 0, i_dup = 0; i < n; ++i, ++i_dup) {
    i_dup = __builtin_opacify(i_dup);
    if (i_dup >= n)
      fault_handler();
    if (a1[i] != a2[i]) {
      if (a1[i_dup] == a2[i_dup])
        fault_handler();
      if (n_dup != n)
        fault_handler();
      return -1;
    }
  }
  if (i_dup < n)
    fault_handler();
  if (n_dup != n)
    fault_handler();
  return 0;
}
  \end{lstlisting}
\end{figure}

To prevent optimizations from removing the redundant data and code, we opacify
every assignment to the duplicated variable (lines~2 and~4 from
\lstRef{lst:loop_blt}): the compiler can no longer detect the identity
relation between the original and its corresponding duplicated variable. Like in the previous example, the resulting opaque chains interleave original computations with checks, and link to a terminating fault handler through a control dependence. We validate the source-level loop hardening scheme on the core loop of PIN authentication \cite{fissc}, named \ic{loop-pin} in the following.

\subsection{Constant-Time Selection}
\label{security:ct}

Another well-known, yet hard to achieve example of security property is
selecting between two values, based on a secret selection bit, in constant time.
This means the generated code for the selection operation must not contain any
jump conditioned by the secret selection bit, otherwise the execution time of
the operation will depend on whether the first or the second value is selected,
thus leaking the secret selection bit. Cryptography libraries resort to
data-flow encoding of control flow, bitwise arithmetic at source level to avoid
conditional branches, but this fragile constant-time encoding may be altered by
an optimizing compiler.

\begin{figure}[h!tb]
  \begin{lstlisting}[xleftmargin=0.35cm]
/// a. Constant-time selection between two values, version 1
uint32_t ct_select_vals_1(uint32_t x, uint32_t y, bool b) {
  signed m = 0 - b;
  return (x & m) | (y & ~m);
}
  \end{lstlisting}

  \begin{lstlisting}[xleftmargin=0.35cm]
/// b. Constant-time selection between two values, version 2
uint32_t ct_select_vals_2(uint32_t x, uint32_t y, bool b) {
  signed m = 1 - b;
  return (x * b) | (y * m);
}
  \end{lstlisting}

  \begin{lstlisting}[xleftmargin=0.35cm, captionpos=b, label={lst:ct},
                     caption={Constant-time selection attempts.}]
/// c. Constant-time selection from lookup table
uint64_t ct_select_lookup(const uint64_t tab[8], const size_t idx) {
  uint64_t res = 0;
  for (size_t i = 0; i < 8; ++i) {
    const bool cond = (i == idx);
    const uint64_t m = (-(int64_t)cond);
    res |= tab[i] & m;
  }
  return res;
}
  \end{lstlisting}
\end{figure}

Consider different functions from \lstRef{lst:ct}. Other than the first two
attempts at implementing constant-time selection between two values \ic{x} and
\ic{y} based on a secret selection bit \ic{b}, we also consider an example where
the programmer wishes to select a value from the lookup table \ic{tab} while
hiding the secret lookup index \ic{idx}. All these functions are carefully
designed to contain no branch conditioned by the secret value: a bitmask \ic{m}
is created from the secret value using arithmetic tricks, then is in turn used
to select the wanted values. Nevertheless, it has been reported that the code
generated by LLVM is not guaranteed to be constant-time. For instance, the
compiler introduces a conditional jump based on the secret value when compiling,
with optimizations enabled, the first two functions for IA-32 \cite{simon}, or
the last function for both IA-32 and x86-64 \cite{ct}. The community is
desperately in search of a reliable way to prevent the compiler from spotting
and optimizing the constant-time idioms. The most certain approach currently
available is perhaps to introduce to the compiler a specially-crafted builtin
that will be ultimately compiled into a conditional move instruction (if
available in the target architecture) \cite{simon}. However, this is rather a
proof-of-concept and lack of generalizability: it only supports the operation of
selecting between two values, and needs to be rewritten in order to implement
the constant-time selection from lookup table from \lstRef{lst:ct}.c for
instance.

We propose in \lstRef{lst:ct_blt} an alternative to the above solution, relying on our opacification
mechanism. The intuition is to hide from the compiler the correlation between
the bitmask \ic{m} and the secret selection bit \ic{b}. This prevents the compiler
from recognizing the selection idioms and turning it into conditional jumps: it embeds bitwise logic into an opaque chain linking select arguments to the return value. Moreover, we do not assume calls to
these constant-time selection functions to be part of opaque chains; instead we create an opaque chain inside each function and make sure that it terminates by an opaque expression by opacifying the return
value. This is a case of conditional transformation-preservation of instructions (Definition~\ref{def:tpi}): individual selection operations may or may not execute depending on (non-sensitive) program input,
but as soon as one of these executes, the constant-time expressions it encloses will be transformation-preserved due to the opaque chain forcing the compiler to compute the bitmask (and its complement) then
using it for the selection. Furthermore, for \lstRef{lst:ct_blt}.c, not only we want to ensure that the compiler does not transform the selection inside the loop into a branch conditioned by the secret
index, but additionally we want to preserve the constant-timeness of the whole loop by making sure that the \ic{|} operation takes place at every iteration. This is implemented by opacifying each element of
the array. It is worth noting that, unlike the traditional approach trying to reliably generate conditional move instruction whenever available, we accurately generate the expected
constant-time code from the programmer's data-flow encoding of control flow. Although this may result in slower code, this can be directly applied to other constant-time operations involving value
preservation. We validate this solution on mbedTLS's RSA decryption \cite{mbedtls} and a self-written RSA exponentiation using Montgomery ladder \cite{simon}, respectively called \ic{ct-rsa} and
\ic{ct-montgomery} in the following.

\begin{figure}[h!tb]
  \begin{lstlisting}[xleftmargin=0.35cm]
/// a. Constant-time selection between two values, version 1
uint32_t ct_select_vals_1(uint32_t x, uint32_t y, bool b) {
  signed m = __builtin_opacify(0 - b);
  return __builtin_opacify((x & m) | (y & ~m));
}
  \end{lstlisting}

  \begin{lstlisting}[xleftmargin=0.35cm]
/// b. Constant-time selection between two values, version 2
uint32_t ct_select_vals_2(uint32_t x, uint32_t y, bool b) {
  signed m = __builtin_opacify(1 - b);
  return __builtin_opacify((x * b) | (y * m));
}
  \end{lstlisting}

  \begin{lstlisting}[xleftmargin=0.35cm, captionpos=b, label={lst:ct_blt},
                     caption={Secure constant-time selections using opacification.}]
/// c. Constant-time selection from lookup table
uint64_t ct_select_lookup(const uint64_t tab[8], const size_t idx) {
  uint64_t res = 0;
  for (size_t i = 0; i < 8; ++i) {
    const bool cond = (i == idx);
    const uint64_t m = __builtin_opacify(-(int64_t)cond);
    res |= __builtin_opacify(tab[i]) & m;
  }
  return __builtin_opacify(res);
}
  \end{lstlisting}
\end{figure}

Interestingly, this does not work as intended for \lstRef{lst:ct_blt}.b: since
\ic{b} is of type \ic{bool}, the compiler knows that \ic{x * b} can only yield
0 or \ic{x}, thus generate a conditional jump by enumerating all possible values
of \ic{b}. This can easily be fixed by slightly modifying the multiplication as
illustrated in \lstRef{lst:ct_fix}; however, this exposes the limit of our
mechanism: we cannot rely on data opacification and dependences to prevent
optimization passes from introducing control-flow to the program. To the best of our knowledge, there
exists no real solution to this problem (yet): it has always been valid for
compilers to modify the program's control-flow as long as this does not alter
the program's behavior, and this is something we usually have no control over.

\begin{figure}[h!tb]
  \begin{lstlisting}[xleftmargin=0.35cm, captionpos=b, label={lst:ct_fix},
                     caption={Secure constant-time selection version 2.}]
uint32_t ct_select_vals_2(uint32_t x, uint32_t y, bool b) {
  signed m = __builtin_opacify(1 - b);
  return __builtin_opacify((x * (1 - m)) | (y * m));
}
  \end{lstlisting}
\end{figure}

\section{Methodology and Validation}
\label{sec:validation}

In this section, we first validate the functional correctness of our observation-preserving approach and implementation, then describe our validation methodology and use it to establish the preservation of the security protections on the applications presented in \secRef{sec:security}.

\subsection{Functional Validation by Checking Value Integrity and Ordering}

Establishing the preservation of an observation event amounts to proving the existence of an observation point at which all observed values are available (cf.\ \defRef{def:opt}), at the proper memory address or associated with the appropriate variable, and following the specified partial order. This consists in checking, for a given program execution, (1) the presence of all observation events, and (2) that for each of these events, the observed values of the specified variables and memory locations are the expected ones, and (3) that event ordering is compatible with the specified partial order. To this end, we leverage the concept of \textit{observation trace}, which is the sequence of program partial states defined by all observation events encountered during a given execution of the program \cite{Vu20}. Practically, validation involves comparing, for a given input of the program, two observation traces:
\begin{enumerate}
\item Reference trace: we execute the reference program compiled with optimizations disabled. The reference program is the original program (without our intrinsics) with \ic{printf} inserted to generate the expected observed values. We assume \ic{-O0} preserves the observation events as well as the partial state of the ISO C abstract machine \cite{C11} containing the observed values of each event.
\item Optimized trace: we execute the program with builtins inserted, compiled with our compilation framework at different optimization levels. We modify a DWARF parser library \cite{pyelftools} to create a list of breakpoints containing all observation addresses in the binary code, as reported in the DWARF section. At each of these addresses, we record the locations where the observed values are stored. We finally retrieve these values during program execution, using a debugger.
\end{enumerate}

To compare the traces, we associate each partial state defined by an observation event (\ic{printf} in the reference program or intrinsic in our version) with a unique identifier---a combination of line and column numbers at which the event is defined in the program source. We then verify using an offline validator (a small Python program) that (1) each partial state in the reference trace has a corresponding counterpart (having the same identifier) in the optimized trace, and inversely (2) each partial state in the optimized trace has a correspondence in the reference trace. The validator also verifies that all values in a given partial state from the optimized trace match the expected values reported in its reference counterpart. Now, this leaves us with the question of validating observation ordering: unlike Vu et al.\ we only enforce a partial ordering on program observation events. There is no particular constraint for the relative order of observation events having no data dependence relation. As a consequence, due to code motion or rescheduling during compiler optimizations, the reference and optimized traces might not be identical. As a result, we propose an incomplete validation methodology aiming for practicality and still providing high confidence. The methodology is twofold: for every benchmark we derive (i) a totally ordered version where a unique temporary variable is used as a written-to and read-from token to chain all observation events (in other words, all events consume and write to the same token), and (ii) a minimally ordered version where a distinct token is produced from every observation event and immediately consumed in a distinct I/O instruction (the latter being decoupled from the former, it does not constrain the ordering of observation events).

Moreover, since we model observation events as side-effect free, pure functions, different observations of the same values may be combined into a single one. As such, during program execution, a single point observing a given value might actually correspond to multiple observations of the same value. Although the transformation is perfectly valid, it leads to false positives reported by the validator when comparing observation traces. To eliminate these erroneous validation results, we update instruction combining support functions in LLVM to embed the metadata representing individual observations to be combined into a single combined observation (referencing both variable names, line numbers, etc.). These embedded observations will eventually be expanded when creating observation breakpoints for the debugger, which allows the corresponding partial states to be logged into the observation trace.

We validate the functional correctness of our implementation on a subset of the test suite of Frama-C, a static analysis framework for C source program \cite{cuoq}. The test suite is designed to validate different Frama-C analyses on a range of C programs representative of the language semantics, using program properties written in the ACSL annotation language \cite{ACSL}. We restrict ourselves to properties verifying the expected values of variables at a given program point, ignoring test cases referring to more advanced ACSL constructs. These properties can easily be expressed as observation events with our intrinsics.

We compile each of these test cases at 6 optimization levels \ic{-O0}, \ic{-O1}, \ic{-O2}, \ic{-O3}, \ic{-Os}, \ic{-Oz}. This results in two sets of 31 applicable test cases featuring 616 observations---one set for totally ordered events (i) and one for minimally ordered events (ii). Notice that these test cases are not meant to be evaluated as performance benchmarks, we only use them to validate the correctness of our implementation.

We automatically verified that in both sets, all 616 properties have been correctly propagated to machine code. In the first set (i) we checked that the observation trace is identical to the reference trace. In the second set (ii) we checked that values are all present and correct, but could not verify any partial ordering constraint (as there is none to be checked). All of this, at all considered optimization levels.

%% AC. Fixing a methodological weakness we talked about but that was apparently lost in further experimental iterations: using a total ordering through a single token variable facilitates trace comparison (equality), but reduces validation coverage. One skeptical reader would notice we did not functionally validate partial orderings. Arguably, it is tough to check partial orderings, but we could at least check existence and values (not trace equality) when relaxing ordering. It should be sufficient to use a separate to token for every observation and to chain it to the I/O consumer in the end (making it consume all tokens). How difficult is it to automate this? AC. Addendum: the decoupled I/O approach is implemented and sufficient.

\subsection{Security Protection Preservation Validation}

Validating the preservation of security protections is more challenging. While verifying value integrity is enough to prove the preservation of observation events, it is only a necessary condition for preserving security protections. Hence, other than applying value integrity verification to the security applications from \secRef{sec:security}, we also define additional mechanisms to validate specific components of the preserved protections.

\paragraph{Checking Value Utilization}
In our security examples, opacification is used to (1) protect key values of the security countermeasure which are subject to program optimizations---such as duplicated variables from the loop hardening scheme or step counters from SCI protection---from being optimized away, and (2) make sure that the protected values are actually used in the subsequent code of the countermeasure (different security checks for instance). Clearly, value integrity verification only guarantees the former, we need a second verification to assess the latter. On the one hand, we first determine the uses of opacified values in the program, then verify that they are indeed part of the original countermeasure in the source program. On the other hand, we first find the critical parts of each of the security protections, such as the removal of the old mask from the masking operation, or various redundant checks. We then determine the operands of these key computations and verify that they indeed are results of opacifications.

The mechanism is undoubtedly protection-dedicated, as important uses of the opacified values really depend on the considered countermeasure scheme, and thus requires manual inspection of the generated code.
Nonetheless, the two-phase verification can be implemented as an automated data-flow analysis at the program's MIR or machine code, as long as the analysis tool knows about uses of the opacified values crucial to the considered countermeasure.
%The verification relies on a forward data-flow analysis, implemented as a late MIR pass scheduled right before
%object/assembly code emission, that determines the uses of all opacified values in the program. These uses are then manually inspected to ensure that they actually correspond to the important uses of the
%opacified values from the original countermeasure in the source program.

\paragraph{Checking Statement Ordering}
For SCI, other than preserving the step counter incrementations and the security checks, we also need to guarantee the proper interleaving of functional and countermeasure statements. This requires opacifying operands of every statement with an artificial dependence of the result from previous statement, thus creating an opaque chain. This ensures that, given two consecutive C statements $S_1$
and $S_2$, all MIR/assembly instructions between the opacification of the first operand of $S_1$ and the opacification of the first operand of $S_2$ correspond to $S_1$. We manually inspect the generated code to verify this ordering.

Notice that there would be an option to automate this verification of the proper interleaving of functional and countermeasure statements, \emph{if} one trusts the debug information to be sound and accurate. We could verify the line numbers, mapping MIR/assembly instructions to the corresponding source statements and vice versa. If (1) all MIR/assembly instructions between the opacification of the first operand of a statement $S_1$ and the opacification of the first operand of the next statement $S_2$ have the same line number, and (2) during the scan over every MIR/assembly instructions of the program, the line number reported for each instruction (from the same basic block) is in an ascending order, it can be concluded that functional and countermeasure statements are correctly interlaced. Unfortunately debug information is not robust enough in general and we prefered to rely on manual inspection for higher confidence.

\paragraph{Checking Constant-Time Selection}
As explained in \ssecRef{security:ct}, a widely-adopted informal definition of constant-time selection is that there is no conditional branch based on a secret selection bit. To validate the preservation of constant-timeness using our opacification approach, we manually verify that none of the following three values is used to compute branch conditions: a secret selection bit, a bitmask created from it, or its opacified value. Notice that a conservative verification scheme could be implemented as a static data-flow analysis on the generated machine code, even though the automated determination of whether the branch conditions depend on the secret selection bit might be challenging, notably when the data flow involves memory accesses.

\paragraph{Validation Results}
\tabRef{tab:validation} summarizes different validation schemes that we apply for each application presented in \secRef{sec:security}, in order to verify the preservation of different security countermeasures. For each application we could verify that the appropriate schemes yield the expected results, validating our approach and implementation.

\begin{table}[h]
  \footnotesize
  \begin{tabular}{|c|c|c|c|c|c|}
\cline{2-6}
%\multicolumn{1}{l|}{}   & \texttt{rsa-\{en|de\}crypt} & \texttt{aes-herbst}/\texttt{rotate-mask} & \texttt{pin-loop} & \texttt{\{aes|pin\}-sci} & \texttt{ct-sel}\\
\multicolumn{1}{l|}{}   & \texttt{erasure-*}             & \texttt{mask-*}                & \texttt{loop-pin}              & \texttt{sci-*}                 & \texttt{ct-*}\\
\hline
Value Integrity         & \cmark {\color{mygreen}\cmark} & \cmark {\color{mygreen}\cmark} & \cmark {\color{mygreen}\cmark} & \cmark {\color{mygreen}\cmark} & \cmark {\color{mygreen}\cmark}\\
\hline                                                         
Value Utilization       & N/A                            & \cmark {\color{mygreen}\cmark} & \cmark {\color{mygreen}\cmark} & \cmark {\color{mygreen}\cmark} & \cmark {\color{mygreen}\cmark}\\
\hline                                                              
Statement Ordering      & N/A                            & N/A                            & N/A                            & \cmark {\color{mygreen}\cmark} & N/A   \\
\hline                                                                                                                                  
Constant-time Selection & N/A                            & N/A                            & N/A                            & N/A                            & \cmark {\color{mygreen}\cmark}\\
\hline                                                                                      
  \end{tabular}
  \caption{Validation of different security applications. \cmark indicates the scheme is \textit{applied} to the program, N/A indicates the opposite. {\color{mygreen}\cmark} indicates the scheme is \emph{validated} for the program.}
  \label{tab:validation}
\end{table}

\section{Experimental Evaluation}
\label{sec:evaluation}

Let us now evaluate the proposed mechanisms on the security applications presented in \secRef{sec:validation}, focusing on performance and compilation time impact.
% applications presented in \secRef{sec:validation}, where we analyze the impact of our mechanisms on the program's performance.

\subsection{Experimental Setup}

For each one of the considered security applications, we first compare our versions against the unoptimized programs---which is also a solution to preserving security protections---to quantify performance benefits.
%alongside the compilation time overhead
We then compare our versions against other available preservation mechanisms, namely compiler-dependent programming tricks for constant-time
selection \cite{simon}, and Vu et al.'s I/O-based approach for all other applications \cite{Vu20}. For fairness purposes we use the same version of LLVM as Vu et al.
Eventually, we also compare our compiler-based implementation with an alternative one that does not involve modifications to Clang and LLVM but relying on inline assembly instead; we will describe it in the
following. Finally, we present the compilation time overhead of our implementation.
%In our evaluation, we refer to the version with property preservation mechanisms as \ic{Annotation}, the original as \ic{Original}, the implementation with inline assembly as \ic{Inline Assembly},
%and our implementation as \ic{Builtin}.

For all benchmarks, we target two different instruction sets: ARMv7-M/Thumb-2 which is representative of deeply embedded devices, and Intel x86-64 representative of high-end processors with a complex micro-architecture.
In addition, since Simon et al.\ showed that the compilation of source-level constant-time selection code on the IA-32 architecture contained secret-dependent conditional jumps, we also consider IA-32 for the constant-time applications \ic{ct-rsa} and \ic{ct-montgomery}.

Performance evaluation for the ARMv7-M/Thumb-2 ISA takes place on an MPS2+ board with a 32-bit Cortex-M3 clocked at 25 MHz with 8 Mb of SRAM, while our Intel test bench has a quad-core 2.5 Ghz Intel Core i5-7200U CPU with 16 GB of RAM.

We use the Intel platform for compiling for either target. Changing the target only concerns the back-end, a short part of the compilation pipeline, as a result, we only report the compilation time evaluation results for the ARMv7-M/Thumb-2 ISA.

Our experiments cover all common optimization levels (\ic{-O1}, \ic{-O2}, \ic{-O3}, \ic{-Os}, \ic{-Oz}). Performance measurements are based on the average of 10 runs of each benchmark and configuration.
%% \figRef{fig:performance_O0} and \figRef{fig:performance_cc}.

\subsection{Comparing to Unoptimized Programs}

\figRef{fig:performance_O0} presents the speed-up of our approach at different optimization levels over unoptimized programs. For all benchmarks, speedup ranges from $1.2$ to $12.6$, with an harmonic mean of $2.8$.
Clearly, our observation- and opacity-based approach to preserving security protections enables aggressive optimizations with significant benefits over \ic{-O0}.

\begin{figure}[h!tb]
    \hspace{-0.7cm}
    \centering
    \includegraphics[width=\textwidth]{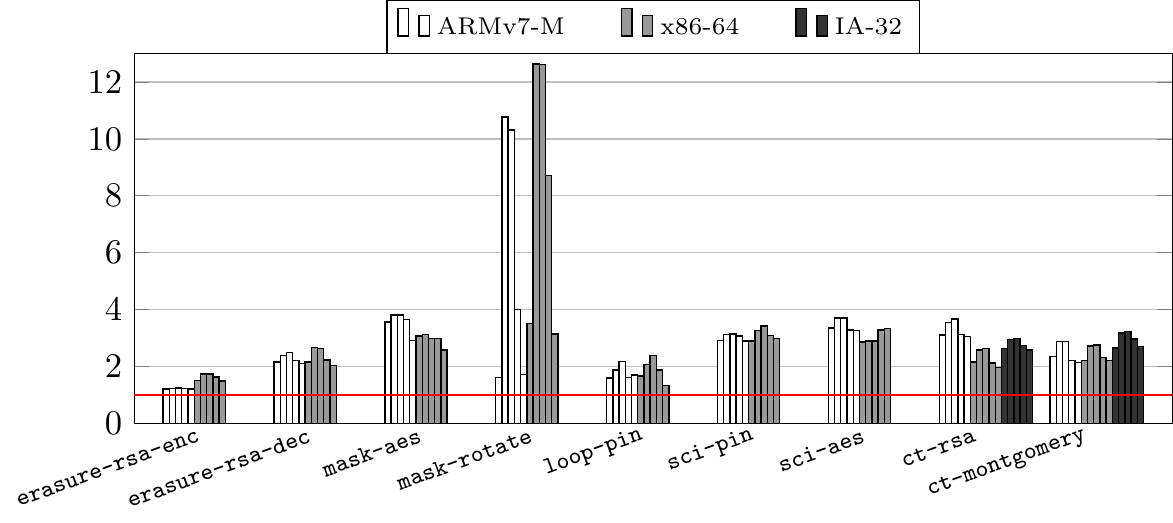}
    \vskip-1.5ex
    \caption{Speed-up of our approach over unoptimized original programs---ordered by compiler option \ic{-O1}, \ic{-O2}, \ic{-O3}, \ic{-Os}, \ic{-Oz}. The red line represents a performance ratio of 1.}
    \label{fig:performance_O0}
\end{figure}

\subsection{Comparing to Reference Preservation Mechanisms}

\begin{figure}[h!tb]
    \hspace{-0.7cm}
    \centering
    \includegraphics[width=\textwidth]{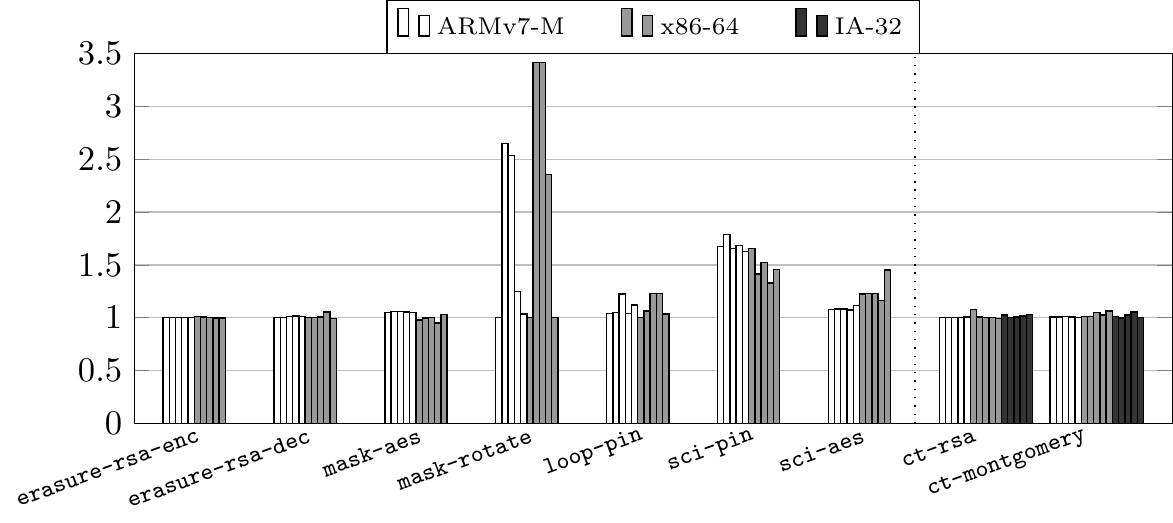}
    \vskip-1.5ex
    \caption{Speed-up of our approach over reference preservation approaches---compilation-time property-preserving annotations \cite{Vu20} for applications on the left side of the dotted line and programming tricks \cite{simon} for the ones on its right side---ordered by compiler option \ic{-O1}, \ic{-O2}, \ic{-O3}, \ic{-Os}, \ic{-Oz}.}
    \label{fig:performance_cc}
\end{figure}
    
\figRef{fig:performance_cc} presents the speedup our approach compared to reference preservation approaches, at different optimization levels.
For \ic{erasure-rsa-enc}, \ic{erasure-rsa-dec}, \ic{mask-aes}, \ic{mask-rotate}, \ic{loop-pin}, \ic{sci-pin} and \ic{sci-aes}, we compare our approach against the property preservation mechanism proposed by Vu et al.\ \cite{Vu20}.
The authors introduced new intrinsics to implement their property preservation mechanism, however, they rely heavily on the I/O side effects of the intrinsics: not only they introduce I/O side-effecting intrinsics to model observation points so that these cannot be removed by optimizations, they also insert I/O side-effecting artificial definitions for every property-observed value to protect these from being optimized out.
Furthermore, in order to guarantee the correct debug information for these values, the authors inserted more artificial definitions to prevent multiple live ranges corresponding to the same source variable from overlapping. Furthermore, to ensure the correct values in memory at observation points, these intrinsics also behave like memory fences, i.e.\ can read from and write to memory.
As a consequence, our implementation with pure intrinsics (no side effects), accessing memory only when required, should enable more optimizations and thus result in faster code.
Our results clearly confirm this. For example, although \ic{rotate-mask} contains the masking computation, the data used in the operation is passed as function arguments instead of being declared as global variables in reference implementations; this clearly allows more optimizations when the function is inlined (i.e.\ when compiled at \ic{-O2}, \ic{-O3} or \ic{-Os}), and especially when the function call is inside a loop. More generally, optimizations such as ``loop unrolling'' and ``loop invariant code motion'' are the main sources of benefits with our approach at these optimization levels. On the contrary, for \ic{erasure-rsa-enc} and \ic{erasure-rsa-dec}, the function implementing the protection only contains the erasure of the sensitive buffer, we thus observe almost no difference compared to the property preservation mechanism proposed by Vu et al.
Similarly for \ic{aes-herbst}, the data required for the masking computation is stored in memory as global variables, there is almost no difference between two versions: the masking operation contains loads and stores to the secret key as well as the different masks in the order defined in the source program, as the protection is correctly preserved. As for other applications, for both targets, we note a clear improvement, ranging from 1.04 to 1.79, with an average of 1.3. Overall, compared to our approach, I/O side-effecting intrinsics restrict compiler optimizations, thus inevitably degrade the performance of generated code.

%Wide adoption of security protection is hampered by the non-negligible
%overhead they usually incur, this justifies the needs for a more lightweight approach of preserving security countermeasures, such as our proposed solution.

For \ic{ct-rsa}, we compare our approach against the constant-time selection implementation of mbedTLS \cite{mbedtls}, which is basically the same as the version from \lstRef{lst:ct}.a, but with the
computation of the bitmask (line~3) splitted into a separate function. Furthermore, this function must not be inlined in order to prevent the compiler from optimizing it away. As for \ic{ct-montgomery}, we
compare our approach against the specially-crafted implementation of OpenSSL \cite{openssl}. It is worth noting that general-purpose compilers offer no guarantees of preserving constant-timeness: future versions of the same compiler may spot the trick and optimize the constant-timeness away \cite{simon}. Although our approach allows the functions implementing constant-time
selection to be inlined while still preserving constant-timeness, these only take a small fraction of the execution time; we do not notice a clear difference compared to other constant-time implementations.

\subsection{Comparing to Alternative Implementations}

Production compilers provide an \emph{inline assembly} syntax to embed target-specific assembly code in a function.
The feature is regularly used by C programmers for low-level optimizations and operating system primitives, and also for sensitive applications to avoid interference from the compiler \cite{rigger}. GCC-compatible compilers implement an extension of the optional ISO C standard syntax for inline assembly, allowing programmers to specify inputs or outputs for inline assembly as well as its behavior w.r.t.\ memory accesses and I/O effects \cite{stallman}. This specification stands as a contract between the assembly code and the compiler. Compilers, relying on the contract, are completely agnostic to what happens inside the an inline assembly region. In other words, the assembly code region is opaque to the compiler. We may thus leverage this feature to implement opaque expressions. For example, to preserve the correct masking order in \lstRef{lst:comp-order-blt}, the call to \ic{\_\_builtin\_opacify} at line 4 may be replaced by an inline assembly expression, as shown in \lstRef{lst:comp-order-asm}.

\begin{figure}[h!tb]
  \begin{lstlisting}[xleftmargin=0.35cm, caption={Secure masking using opacification based on inline assembly.},
                     captionpos=b, label={lst:comp-order-asm},
                     emph={__asm__}]
k ^= m;
...
uint8_t tmp = k ^ mpt;
__asm__ ("" : "+r" (tmp));
k = tmp ^ m;
  \end{lstlisting}
\end{figure}

The inline assembly region (line 4) is actually empty: the behavior exposed to the compiler of the whole expression is specified by the \ic{"+r" (tmp)} constraint. This
essentially means that \ic{tmp} is both the input and output of the expression, and that the expression neither accesses memory nor does it have any side-effect. As an output of the inline assembly expression, \ic{tmp} is now opaque to the compiler, just as if it was defined by the \ic{\_\_builtin\_opacify}.

This example can be generalized to implement any opaque region. In practice, it is sufficient to implement a small set of builtins covering the typical opacification scenarios. A set of preprocessor macros can be designed to cover these typical scenarios and provided as a portable interface across most compilers and targets.

Note that this approach slightly complicates the implementation of observations, carrying precise variable names, memory addresses, line numbers down to machine code. Additional conventions and post-pass on the generated assembly code are required to produce the appropriate DWARF representation, as described in \secRef{sec:implementation}.

Now, the natural question is to compare the performance of an inline-assembly-based implementation with our compiler-native opaque regions. To this end, we consider a subset of the applications presented in \secRef{sec:security}, containing \ic{erasure-rsa-enc}, \ic{erasure-rsa-dec}, \ic{mask-aes}, \ic{mask-rotate}, \ic{loop-pin}, \ic{ct-rsa} and \ic{ct-montgomery}. We exclude \ic{sci-pin} and \ic{sci-aes}, as these applications would require the manual insertion of inline assembly expressions at every statement of C source programs, which is impractical.
\figRef{fig:performance_asm} presents the speedup our compiler-native implementation w.r.t.\ inline-assembly, at different optimization levels.

\begin{figure}[h!tb]
    \hspace{-0.7cm}
    \centering
    \includegraphics[width=\textwidth]{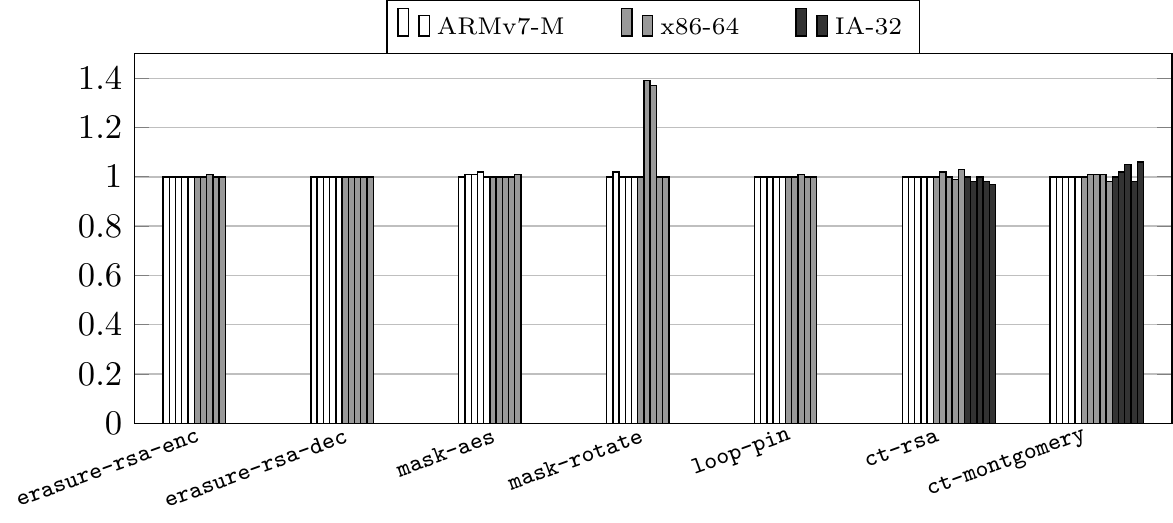}
    \vskip-1.5ex
    \caption{Speed-up of our approach over inline assembly---ordered by compiler option \ic{-O1}, \ic{-O2}, \ic{-O3}, \ic{-Os}, \ic{-Oz}.}
    \label{fig:performance_asm}
\end{figure}

In the majority of cases, the two considered implementations generate the same executable code. For \ic{ct-rsa} and \ic{ct-montgomery}, there is a slight difference in performance due to discrepancies in register allocation. This is not visible in other applications because these two programs are larger and demonstrate higher register pressure. The only significant performance difference is for \ic{mask-rotate} compiled with \ic{-O2} and \ic{-O3} for x86-64: our compiler-native implementation is $40\%$ faster than inline assembly. The core loop of the inline-assembly-based version happens not to be unrolled, while compiler-native version's is. Interestingly, this is only the case for x86-64: the same loop is unrolled for both versions when compiling for ARMv7-M. Indeed, the difference disappears when we force loop unrolling using the \ic{-funroll-loops} option together with \ic{\#pragma unroll}. As expected, inline assembly occasionally interferes with compiler optimizations, despite the precise specification enabled in its syntax, while compiler intrinsics allow for carrying more precise semantics to the optimizers. Mitigations exist, and make the inline assembly approach interesting to some multi-compiler development environments. The take away from this is that both approaches are sound and leverage the same formalization and secure development scenarios (for opacification purposes, not for observation purposes). Yet this may not always be the case in the future: compilers are not forbidden to analyze inline assembly and take optimization decisions violating the opacity hypothesis; the fact they do not do it today is no guarantee that secure code will remain secure in future versions. On the contrary, our intrinsics in the source language and IR have a explicit and future-proof opacification and observation semantics.

\subsection{Compilation Time Overhead}

\figRef{fig:compilation_time} shows the compilation time overhead compared to compiling the original programs at the same optimization level. Note that the optimized original programs are insecure, as protections have been stripped out or altered by optimizations. We consider the Intel platform since it is used for both native and cross-compilation for both targets.

\begin{figure}[h!tb]
    \hspace{-0.7cm}
    \centering
    \includegraphics[width=\textwidth]{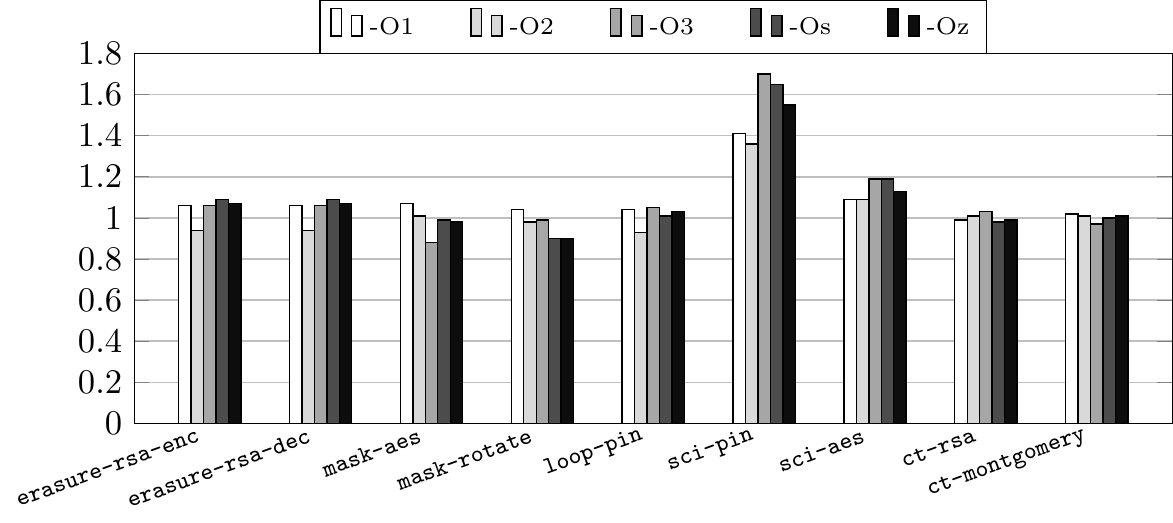}
    \vskip-1.5ex
    \caption{Compilation time overhead on the Intel platform, compared to compiling the original programs at the same optimization level.}
    \label{fig:compilation_time}
\end{figure}

In general, the compilation time overhead is under 10\%, except for \ic{sci-aes} and \ic{sci-pin} where it ranges from 13\% up to 70\%. As discussed in \ssecRef{security:sci}, the SCI protection represents a very important part of the whole application (as counter incrementations are inserted after each C instruction), and it is completely stripped out from original programs when optimizations are enabled without opacification. As a consequence, the code size of the insecure baseline is much smaller than the secure code with fully-preserved countermeasures, which justifies the important compilation time difference.

\section{Conclusion}

We formally defined the notion of observation and its preservation through program transformations. We instantiated this definition and preservation mechanisms through multiple program representations, from C source code down to machine code. The approach relies on two fundamental principles of compiler correctness: (1) the preservation of I/O effects and (2) the interaction of data dependences with program constructs that are opaque to static analyses. We formally proved the correctness of the approach on a simplified intermediate language, and validated it within the LLVM framework with virtually no change to existing compilation passes. Our proposal specifically addresses a fundamental open issue in security engineering: preserving security countermeasures through optimizing compilation flow. Avenues for further research include software engineering scenarios such as testing of production code and more robust debugging of optimized code.

\bibliographystyle{ACM-Reference-Format}
\bibliography{ref}

%% \input{src/appendix}

\end{document}